\documentclass[11pt]{article}

\usepackage[top=1in, bottom=1in, left=0.8in, right=0.8in]{geometry}

\usepackage{amsmath,amsthm,amssymb} 
\usepackage{mathtools}
\usepackage{tikz-cd}
\usepackage{graphicx}
\usepackage{fancybox}
\usepackage{hyperref}
\usepackage{varwidth}
\usepackage{mdframed}
\usepackage{mathrsfs}
\usepackage{physics}
\usepackage{dsfont}
\usepackage{xcolor}
\usepackage{subcaption}
\usepackage{float}

\hypersetup{
    colorlinks,
    linkcolor={red!50!black},
    citecolor={blue!50!black},
    urlcolor={blue!80!black}
}
\DeclareSymbolFont{usualmathcal}{OMS}{cmsy}{m}{n}
\DeclareSymbolFontAlphabet{\mathcal}{usualmathcal}

\usepackage{helvet}

\AtBeginDocument{
  \DeclareSymbolFont{AMSb}{U}{msb}{m}{n}
  \DeclareSymbolFontAlphabet{\mathbb}{AMSb}}

\newcommand{\zz}{\mathbb Z}   
\newcommand{\rr}{\mathbb R}   
\newcommand{\nn}{\mathbb N}   
\newcommand{\cc}{\mathbb C}   
\newcommand{\cH}{\mathcal H} 
 
\newcommand{\cB}{\mathcal B}  
  
\newcommand{\cA}{\mathcal A}  
\newcommand{\cE}{\mathcal E}  
\newcommand{\cK}{\mathcal K}  
\newcommand{\cO}{\mathcal O}  
\newcommand{\rmEnd}{\mathrm{End}}

\newcommand{\al}{\alpha}
\newcommand{\bet}{\beta}  
\newcommand{\ga}{\gamma}  
\newcommand{\ds}{\mathds{1}}  
  
\newcommand{\lpline}[1]{\draw[gray!55] (0.55,0) -- ({#1+0.45},0);}
\newcommand{\lpsites}[1]{\lpline{#1}\foreach \i in {1,...,#1}{\fill (\i,0) circle (1.7pt);}}
\newcommand{\lplabels}[1]{\foreach \i in {1,...,#1}{\node[font=\tiny,gray!75] at (\i,-0.34) {\i};}}
\newcommand{\lparc}[2]{\draw[thick] (#1,0) .. controls (#1,{0.5*(#2-#1)+0.30})
    and (#2,{0.5*(#2-#1)+0.30}) .. (#2,0);}
\newcommand{\lpthru}[1]{\draw[thick] (#1,0) -- (#1,1.05);}

\renewcommand{\epsilon}{\varepsilon}
\renewcommand{\phi}{\varphi}
\renewcommand{\emptyset}{\varnothing}
\renewcommand{\geq}{\geqslant}
\renewcommand{\leq}{\leqslant}

\renewcommand{\Im}{\operatorname{Im}}

\numberwithin{equation}{section}

\newtheorem{thm}{Theorem}[section]
\newtheorem{lem}[thm]{Lemma}
\newtheorem{prop}[thm]{Proposition}
\newtheorem{cor}[thm]{Corollary}
\newtheorem{definition}[thm]{Definition}

\theoremstyle{remark}
\newtheorem{remark}[thm]{Remark}
\theoremstyle{plain}

\begin{document}

\begin{center}{\Large \textbf{
On the Spectrum of Some Temperley--Lieb Spin Chains\\
}}\end{center}

\begin{center}\textbf{
Robert Ferydouni\textsuperscript{1,2$\star$} and
Bruno Nachtergaele\textsuperscript{1,2$\dagger$}
}\end{center}

\begin{center}
{\bf 1} Department of Mathematics, University of California Davis, Davis, CA 95616, USA
\\
{\bf 2} Center for Quantum Mathematics and Physics (QMAP), University of California Davis, Davis, CA 95616, USA
\\[\baselineskip]
$\star$ \href{mailto:rferydouni@ucdavis.edu}{\small rferydouni@ucdavis.edu}\,,\quad
$\dagger$ \href{mailto:bnachtergaele@ucdavis.edu}{\small bnachtergaele@ucdavis.edu}
\end{center}

\section*{Abstract}
We review mathematical results about the spectrum of spin-$s$ chains with a nearest neighbor interaction given by the projection onto the singlet state of two spin-$s$ degrees of freedom. Up to normalization, these interaction terms generate a Temperley--Lieb algebra at loop parameter $d=2s+1$. We also present several new results. First, we prove a ferromagnetic ordering of energy levels across symmetry sectors of the model, thus generalizing a previous result for the spin-$\tfrac12$ $XXZ$ chain \cite{nachtergaele2004ferromagnetic}. This enables us to determine a uniform lower bound on the finite volume spectral gap for these spin chains. By combining a variational estimate with Knabe's finite size criterion, we show that the exact spectral gap of the GNS Hamiltonian of the fully polarized state is $1-2/d$. Furthermore, this value serves as a lower bound for the gap of the maximally mixed ground state. Finally, we show that the ground state projectors on open chains satisfy the Jones--Wenzl recursion, thereby recovering a known result that the degeneracies are the quantum integers $[L+1]_q$ for $q+q^{-1}=d$. We discuss how the high degeneracy of the ground states allows the model to break a continuous symmetry while remaining gapped, and renders that gap unstable.

\vspace{\baselineskip}

\vspace{10pt}
\noindent\rule{\textwidth}{1pt}
\tableofcontents
\noindent\rule{\textwidth}{1pt}
\vspace{10pt}

\section{Introduction}
\label{sec:intro}

The Temperley--Lieb algebra was introduced in the study of the Potts and ice-type models \cite{temperley1971relations} and has since become one of the standard organizing structures of one dimensional quantum magnetism. Its role there is unusual: it is not a symmetry of any particular chain but an algebra in which many different chains simultaneously live. If the nearest neighbour interaction terms of a spin chain satisfy the Temperley--Lieb relations at a common loop parameter, then the Hamiltonian is, as an abstract algebra element, the same in every such model. If the representation is faithful, its spectrum is determined by the algebra rather than by the local Hilbert space in which the algebra happens to be realized. What is \emph{not} determined this way are the multiplicities, and it is in the gap between these two statements that the models differ.

This paper studies the family of chains whose interaction is the orthogonal projection $P^{(0)}_{x,x+1}$ onto the singlet state of two neighbouring spin-$s$ particles at $\{x,x+1\}$. We refer to these chains collectively as the singlet model throughout the rest of the paper. These terms generate a Temperley--Lieb algebra at loop parameter $d = 2s+1$, and the resulting Hamiltonians are frustration free and positive semidefinite with the fully polarized state among their ground states. For $s = \tfrac12$ the model is the isotropic ferromagnetic Heisenberg chain; for $s = 1$ it is, up to a shift and a rescaling, the purely biquadratic chain $\sum_x (\vec S_x \cdot \vec S_{x+1})^2$, which sits at the boundary of the ferromagnetic phase of the bilinear-biquadratic family.

That family, $\cos\theta\,\left(\vec S_x\cdot\vec S_{x+1}\right) + \sin\theta\,\left(\vec S_x\cdot\vec S_{x+1}\right)^2$, whose conjectured phase diagram is shown in Figure~\ref{fig:Spin1 Diagram}, has been studied intensively. The Uimin--Lai--Sutherland models at $\theta = \tfrac\pi4, \tfrac{5\pi}{4}$ were shown to be integrable, gapless and Bethe ansatz solvable \cite{uimin1970one,lai1974lattice,sutherland1975model}. The same goes for $\theta = -\tfrac\pi4$, solved by Takhtajan and Babujian \cite{takhtajan1982, babujian1982}. The valence bond solid point $\theta = \arctan\tfrac13$ of Affleck, Kennedy, Lieb and Tasaki (AKLT) \cite{affleck1987rigorous, affleck1988valence} provided the first rigorous example of a gapped isotropic chain. The point $\theta = -\tfrac\pi2$ and its higher spin analogs were argued to exhibit dimerization in \cite{affleck1990exact, xian1993exact}. The spectral gap and dimerization were shown rigorously in \cite{warzel2020dimers} and \cite{bjornberg2021dimerization}. The family's phase diagram as a function of $\theta$ was first mapped by Sólyom \cite{solyom1987competing} from finite chain spectra. The AKLT point is contained in the topologically nontrivial Haldane phase $\theta \in \left(-\tfrac\pi4, \tfrac\pi4\right)$. Dimerization is expected to hold in $\theta \in \left(\tfrac{5\pi}{4}, \tfrac{7\pi}{4}\right)$. Ferromagnetism holds for $\theta \in \left(\tfrac{\pi}{2}, \tfrac{5\pi}{4}\right)$. The region $\theta \in \left(\tfrac\pi4, \tfrac\pi2\right)$ is expected to be critical and gapless \cite{itoi1997extended}, and characterized by soft modes at momenta $\pm \tfrac{2\pi}{3}$ \cite{fath1991period}. The Uimin--Lai--Sutherland model, which sits at $\theta=\tfrac\pi4$, and the singlet model, which up to a shift and normalization sits at  $\theta = \tfrac\pi2$, occupy the edges of this region. The critical phase, the ferromagnetic phase, and the role the singlet model plays in between them are discussed further in Section~\ref{sec:instability}.

The spin-$1$ case of the singlet model has the longest history given the literature just discussed. Parkinson attacked the spin-$1$ singlet model by a coordinate Bethe ansatz \cite{parkinson1988}, introducing the one- and two-site deviate states that we use below to construct an explicit band of excitations. Barber and Batchelor \cite{barber1989spectrum} observed that the biquadratic interaction, the nine state Potts model and the spin-$\tfrac12$ $XXZ$ chain all furnish representations of the same Temperley--Lieb algebra, and used this to transfer spectral information between them. Kl\"umper studied the model by inversion relations and by the quantum transfer matrix \cite{klumper1989new, klumper1990investigation, klumper1993thermodynamics}.

Bethe ansatz methods have been applied to Temperley--Lieb chains well beyond spin one. K\"oberle and Lima-Santos solved the deformed biquadratic spin-$1$ chain by a coordinate Bethe ansatz \cite{koberle1994deformed}, and subsequently treated the $A$--$D$ Temperley--Lieb models \cite{koberle1995AD} and graded Temperley--Lieb models \cite{limasantos1998exact} by the same technique, in each case constructing the eigenstates from excitations of a reference state closely related to the vectors of Section~\ref{sec:upper}. These works argue that the spectra of a large class of Temperley--Lieb chains coincide up to degeneracies, a statement we recover for open chains in Proposition~\ref{prop:isospectral} by purely representation theoretic means. 

The symmetry algebra of these chains, that is the commutant of the Temperley--Lieb action, has been identified twice by independent routes. Kulish \cite{kulish2003spin} and Kulish, Manojlovic and Nagy \cite{kulish2008quantum} approach it through the $R$-matrix formalism, showing that the relevant antisymmetrizers vanish and deriving in the process a recursion for the dimension of the space annihilated by all interaction terms, that is for the ground state degeneracy of the models considered here. Read and Saleur \cite{read2007enlarged} arrive at the same enlarged symmetry algebra from the loop model side, give explicit generators for it, and describe its irreducible representations in terms of configurations of nested parentheses and dots; it is their description that we use in Section~\ref{sec:reptheory}, where the parentheses become the arcs of a link pattern and the dots carry a Jones--Wenzl projected vector. The multiplicity structure for periodic chains, which is considerably more delicate, was settled by Aufgebauer and Kl\"umper \cite{aufgebauer2010quantum}, who also derive a recursion for the ground state degeneracy on open and periodic chains.

Our motivation for revisiting these models is that the spin-$\tfrac12$ $XXZ$ chain with anisotropy $\Delta$, which forms a Temperley--Lieb chain at loop parameter $2\Delta$, is understood in far greater detail than its higher spin relatives. Koma and Nachtergaele computed the spectral gap of the ferromagnetic $XXZ$ chain exactly \cite{koma1997spectral} and classified its infinite volume ground states, which include interface states in addition to the polarized ones \cite{koma1998complete}. Nachtergaele, Spitzer and Starr proved that the energy levels of that chain are ordered by total spin, the lowest energy in a sector decreasing as the spin increases \cite{nachtergaele2004ferromagnetic}, a property subsequently extended to $U_q(\mathfrak{sl}_2)$ symmetric chains \cite{nachtergaele2012ferromagnetic}. This type of statement originates with Lieb and Mattis in \cite{liebmattis1962ordering}, who showed the antiferromagnetic Heisenberg chain (and some generalizations) had its lowest energy in a spin sector increasing as the total spin increases. Since the spectrum is an invariant of the algebra, one expects analogues of these results for every $s$; since the multiplicities are not, one expects the analogues to be statements about different sectors. Both expectations are borne out below.

The paper is organized as follows. Section~\ref{sec:notation} fixes notation, records the equivalent formulations of the model, and the sense in which they are isospectral. Section~\ref{sec:foel} proves a ferromagnetic ordering of energy levels, with the sectors labeled not by the $SU(2)$ total spin, which we show cannot work for $s \geq 1$, but by the sectors given from the double centralizer theorem, Theorem~\ref{thm:double_centralizer}, or equivalently by the number of arcs in a link pattern. Section~\ref{sec:gap} computes the spectral gap of the GNS Hamiltonian of the fully polarized state exactly, combining an explicit band of excitations with Knabe's finite size criterion \cite{knabe1988energy}. It also constructs a second gapped infinite volume ground state, the maximally mixed state on the ground state spaces, and computes its mean entropy and mutual information. Section~\ref{sec:projectors} shows that the ground state projections of the open chains satisfy the Jones--Wenzl recursion \cite{wenzl1987sequences}, from which the exponential degeneracy in terms of quantum integers $[L+1]_q$ follows, where we parametrize $q+q^{-1} = d$. Section~\ref{sec:discussion} discusses the consequences of that exponential degeneracy: the model breaks a continuous symmetry while remaining gapped, evading the conclusion of the Goldstone theorem \cite{landau1981energy}, and its gap is unstable.

\section{Notation and Main Results}
\label{sec:notation}

\subsection{The Model and its Equivalent Formulations}
\label{sec:model}

Throughout, $s \in \{1, \tfrac32, 2, \tfrac52, 3, \tfrac72, \dots\}$ denotes a fixed spin value (definitions and results extend to $s=\tfrac12$ when explicitly stated) and $d=2s+1$ the dimension of the on-site Hilbert space $V_s \cong \cc^{d}$. We write $\{\ket{\al} : \al = -s,\dots,s\}$ for the standard basis of $V_s$ diagonalizing $S^3$, and for a chain of length $L$ we set the physical Hilbert space
\begin{equation}\label{eq:hilbert_space}
    \cH_L \coloneq V_s^{\otimes L} = \underbrace{\cc^{d} \otimes \cdots \otimes \cc^{d}}_{L}\,.
\end{equation}
Operators carrying site labels, such as $P^{(0)}_{x,x+1}$ below, are understood to act as the identity on all unlabelled tensor factors.

The decomposition of a pair of neighbouring spins into irreducible $\mathfrak{su}(2)$ representations reads $V_s \otimes V_s = V_{2s} \oplus \cdots \oplus V_0$. The singlet space $V_0$ is one-dimensional, and we fix once and for all the \emph{unnormalized} singlet vector that spans this space as
\begin{equation}\label{eq:singlet_vector}
    \ket{\psi} \coloneq \sum_{\al = -s}^{s} (-1)^{s-\al}\ket{\al, -\al}\,, \qquad \braket{\psi}{\psi} = d\,,
\end{equation}
so that the orthogonal projection of $V_s \otimes V_s$ onto $V_0$ is
\begin{equation}\label{eq:projector_P0}
    P^{(0)} = \frac{1}{d}\ket{\psi}\bra{\psi} = \frac{1}{d} \sum_{\al, \bet = -s}^{s} (-1)^{2s- \al - \bet}\ket{\al, -\al}\bra{\bet, -\bet}\,.
\end{equation}
When we need to record the sites on which \eqref{eq:singlet_vector} is supported we write $\ket{\psi}_{x,y}$, with the convention that the first tensor slot is the site $x$ and the second is the site $y$. Note the reversal rule
\begin{equation}\label{eq:reversal}
    \ket{\psi}_{y,x} = (-1)^{2s}\ket{\psi}_{x,y}\,,
\end{equation}
which reflects the exchange parity of the spin-$s$ singlet and will be responsible for the distinction between integer and half-odd integer spin in Section~\ref{sec:foel}.

For $L \geq 2$ the singlet model on the open chain $[1,L]$ is the Hamiltonian
\begin{equation}\label{eq:hamiltonian_obc}
    H_L \coloneq \sum_{x=1}^{L-1} P^{(0)}_{x,x+1}\,,
\end{equation}
and on the periodic chain of the same length it is
\begin{equation}\label{eq:hamiltonian_pbc}
    H_L^{\text{PBC}} \coloneq \sum_{x=1}^{L} P^{(0)}_{x,x+1}\,, \qquad L+1 \equiv 1\,.
\end{equation}
Each summand is an orthogonal projection, hence $H_L, H_L^{\text{PBC}} \geq 0$. The fully polarized vector $\ket{s\cdots s}$ satisfies $P^{(0)}_{x,x+1}\ket{s \cdots s} = 0$ for every $x$, because $\ket{s, s}$ is the highest weight vector of $V_{2s}$ and therefore has no singlet component. Consequently $\ket{s \cdots s} \in \ker H_L \cap \ker H_L^{\text{PBC}}$. In other words, the model is frustration free with ground state energy $0$.

\subsubsection*{The Temperley--Lieb structure}

The algebraic mechanism behind every result in this paper is that the interaction terms of \eqref{eq:hamiltonian_obc} generate a Temperley--Lieb algebra.

\begin{definition}[Temperley--Lieb algebra]\label{def:TL}
For $L \geq 2$ and $\lambda \in \cc$, the Temperley--Lieb algebra $TL_L(\lambda)$ at loop parameter $\lambda$ is the unital associative $\cc$-algebra with generators $e_1,\dots,e_{L-1}$ and relations
\begin{equation}\label{eq:TL_relations}
    e_x^2 = \lambda e_x\,, \qquad e_x e_{x\pm 1}e_x = e_x\,, \qquad e_xe_y = e_ye_x \ \text{ for } \abs{x-y} > 1\,,
\end{equation}
where $1 \leq x, y, x\pm 1 \leq L-1$.
\end{definition}

The assignment $e_x \mapsto d\,P^{(0)}_{x,x+1}$ extends to a representation $\rho_d : TL_L(d) \to \cB(\cH_L)$ at loop parameter $\lambda = d = 2s+1$. In particular
\begin{equation}\label{eq:hamiltonian_TL}
    H_L = \frac{1}{d}\sum_{x=1}^{L-1}\rho_d(e_x)\,.
\end{equation}
We fix the convention that lowercase $e_1,\dots,e_{L-1}$ always denote the abstract generators of $TL_L(d)$ of Definition~\ref{def:TL}, and any statement written in terms of them is a statement about the algebra. When a specific realization is intended we write
\begin{equation}\label{eq:Ex}
    E_x \coloneq \rho_d(e_x) = d\,P^{(0)}_{x,x+1} \in \cB(\cH_L)\,,
\end{equation}
so that $H_L = \tfrac1d\sum_{x=1}^{L-1}E_x$. The subscript $d$ on $\rho_d$ records that the representation, and with it the loop parameter, is fixed by the spin through $d = 2s+1$; general representations of $TL_L(d)$, not necessarily this one, are denoted $\rho$ and $\mu$. Parameterizing $d=q+q^{-1}$ we have that $q$ is not a root of unity since $d \geq 3$, so $TL_L(d)$ is semisimple by \cite{ridout2014standard}. The same reference actually shows the algebra is semisimple in the case $q=1$ ($d=2$) as well by extra arguments. 

\subsubsection*{Equivalent formulations}

The model admits several presentations that we record here, not all of which are used below. First, the operators $\ds, \vec{S}\cdot\vec{S}, \dots, (\vec{S}\cdot\vec{S})^{2s}$ form a basis of the commutant of $V_s \otimes V_s$ as an $\mathfrak{su}(2)$ representation, so $P^{(0)}$ is a polynomial of degree $2s$ in $\vec{S}\cdot\vec{S}$. For $s=1$, up to normalization and a shift of the spectrum, $P^{(0)}$ is the purely biquadratic interaction $(\vec{S}\cdot\vec{S})^2$; for $s=\tfrac32$ one has
\begin{equation}\label{eq:P0_poly32}
    P^{(0)} = \tfrac{33}{128}\ds + \tfrac{31}{96}\,\left(\vec{S}\cdot\vec{S}\right) - \tfrac{5}{72}\,(\vec{S}\cdot\vec{S})^2 - \tfrac{1}{18}\,(\vec{S}\cdot\vec{S})^3\,.
\end{equation}

The next presentation is in terms of the $O(d)$ singlet. We can relabel the basis of $\cc^d$ given above as $\{e_\alpha\}_{\alpha=-s}^{s}$, and write the $O(d)$ invariant singlet vector as
\begin{equation}\label{eq:Od_singlet}
    \ket{\phi} \coloneq \sum_{\alpha=-s}^{s} e_\alpha\otimes e_\alpha\,,
    \qquad \braket{\phi} = d\,,
\qquad Q \coloneq \tfrac{1}{d}\ket{\phi}\bra{\phi}\,.
\end{equation}
The invariance comes from the fact $\ket{\phi}$ is invariant under $R\otimes R$ for every real orthogonal $R$ on $\cc^d$. Define the matrix $V$ on $\cc^d$ as in \cite[Appendix A]{bjornberg2021dimerization}:
\begin{equation}\label{eq:V_def}
    V e_\alpha \coloneq (-1)^{s-\alpha} e_{-\alpha}\,,
    \qquad\text{equivalently}\qquad
    V_{\beta\alpha} = (-1)^{s-\alpha}\delta_{\beta,-\alpha}\,.
\end{equation}
Then $V$ is real, orthogonal, and
\begin{equation}\label{eq:V_FS}
    V^{T} = (-1)^{2s} V\,, \qquad V^{2} = (-1)^{2s}\ds\,.
\end{equation}
The definition gives for any $s$,
\begin{equation}\label{eq:V_maps_Omega}
(\ds\otimes V)\ket{\phi} = \sum_{\alpha}(-1)^{s-\alpha}\,e_\alpha\otimes e_{-\alpha} = \ket{\psi}\,,
\qquad
(V\otimes \ds)\ket{\phi} = (-1)^{2s}\ket{\psi}\,,
\end{equation}
where the second identity follows from the first after the substitution $\beta = -\alpha$ and the fact $(-1)^{s-\beta} = (-1)^{2s} \cdot (-1)^{s+\beta}$. So a single factor of $V$, placed on either site, carries the $O(d)$ singlet to the $SU(2)$ singlet up to a sign. For integer $s$ one can do better and use the same unitary on both sites. Indeed, $(W\otimes W)\ket{\phi} = \ket{\psi}$ holds for the unitary $W$ as in \cite[Eq.~(1.4)]{bjornberg2021dimerization}:
\begin{equation}\label{eq:W_explicit}
    W e_0 = i^{s} e_0\,,\qquad
    W e_{\alpha} = \tfrac{i^{s-\alpha}}{\sqrt{2}}\bigl(e_{\alpha}+e_{-\alpha}\bigr)\,,\qquad
    W e_{-\alpha} = \tfrac{i^{s-\alpha+1}}{\sqrt{2}}\bigl(e_{\alpha}-e_{-\alpha}\bigr)\,,
    \qquad \alpha > 0\,.
\end{equation}
Then
\begin{equation}\label{eq:two_site_equiv}
    (\ds\otimes V)\,Q\,(\ds\otimes V^T) = P^{(0)} = (V\otimes\ds)\,Q\,(V^{-1}\otimes\ds)
    \qqtext{and}
    (W\otimes W)\,Q\,(W\otimes W)^{\ast} = P^{(0)}
\end{equation}
for all $s$ and all integer $s$ respectively. These are statements about a single bond interaction term, but we need unitary transformations that identify the spectra of the Hamiltonians on the full chain. Writing $H_L^{\,\phi} \coloneq \sum_{x=1}^{L-1} Q_{x,x+1}$ for the open $O(d)$ singlet chain, we have the unitary equivalence
\begin{equation}\label{eq:Od_equiv}
    W^{\otimes L}\,H_L^{\,\phi}\,(W^{\ast})^{\otimes L} = H_L\,,
\end{equation}
and for periodic boundary conditions as well. For half-odd integer $s$ no uniform choice is available and we must place $V$ on one site of each bond; if the lattice is bipartite with vertex classes $A$ and $B$, then setting $\mathcal{V} \coloneq \bigotimes_{x\in A} V_x \otimes \bigotimes_{x\in B}\ds_x$ every edge sees either $V\otimes\ds$ or $\ds\otimes V$, whence $\mathcal{V}\,H_L^{\,\phi}\,\mathcal{V}^{-1} = H_L$. In particular the $SU(2)$ and $O(d)$ singlet models are unitarily equivalent for every $s$ on the open chain, and on the periodic chain for $L$ even. For half-odd integer $s$, an odd length periodic chain admits no two-coloring and the argument gives nothing there.

The next presentation is in terms of the $SU(d)$ singlet which requires no further work, since the $SU(d)$ singlet is the same vector $\ket{\phi}$. Write $\mathbf{d}$ and $\overline{\mathbf{d}}$ for the fundamental and antifundamental representations of $SU(d)$. Then the tensor representation decomposes into the adjoint and singlet representation as $\mathbf{d}\otimes\overline{\mathbf{d}} = (\mathbf{d^2-1})\oplus\mathbf{1}$. Using
$g^{T}=\overline{g^{-1}}$ for $g \in SU(d)$, one has
\begin{equation}\label{eq:SUd_invariance}
\begin{split}
    (g\otimes\bar g)\ket{\phi}
    &\;=\; \sum_{\al} g\,e_\al \otimes \bar g\,e_\al
    \;=\; \sum_{\al,\bet} g_{\bet\al}\, e_{\bet}\otimes \bar g\,e_\al \\[4pt]
    &\;=\; \sum_{\bet} e_{\bet}\otimes \bar g\Big(\sum_{\al} g_{\bet\al}e_\al\Big)
    \;=\; \sum_{\bet} e_{\bet}\otimes \bar g\, g^{T} e_{\bet}
    \;=\; \ket{\phi}\,,
\end{split}
\end{equation}
so $\ket{\phi}$ spans the singlet of $\mathbf{d}\otimes\overline{\mathbf{d}}$ and
$Q$ is the projection onto it. The alternating $SU(d)$ chain is therefore the same Hamiltonian $H^{\,\phi}_L$. \eqref{eq:Od_equiv} and its staggered counterpart already give its unitary equivalence to the $SU(2)$ singlet model; only the group one lets act distinguishes the two descriptions.

Since each interaction term was defined as the projection onto an $SU(2)$ representation, the Hamiltonian manifestly has $SU(2)$ symmetry. The other formulations show there is a higher $SU(d)$ and $O(d)$ symmetry. The ground state spaces form representations of these symmetry groups, though reducible for $L \geq 3$. The correct way to decompose the spectrum of the model is discussed below when we use the double centralizer theorem, which generalizes Schur--Weyl duality to spins $s > \tfrac12$. 

The last presentation we record differs in kind from the previous ones. It is a realization of $TL_L(d)$ on a Hilbert space that is not $\cH_L$, so it is not a rewriting of the singlet model but a second model built on the same algebra; it will be our reference chain in Section~\ref{sec:foel_all}. Recall $d=q+q^{-1}$ for $q > 1$, or equivalently, $q = \tfrac12\big(d+\sqrt{d^2-4}\,\big)$, and put $\Delta = d/2$. Let $\mu_d : TL_L(d) \to \cB\big((\cc^2)^{\otimes L}\big)$ be determined by
\begin{equation}\label{eq:xxz_kink_interactions}
    \mu_d(e_x) = \ket{\psi_q}\bra{\psi_q}
\end{equation}
where $\ket{\psi_q} = q^{1/2}\ket{\uparrow \downarrow} - q^{-1/2}\ket{\downarrow \uparrow}$ is the $U_q(\mathfrak{sl}_2)$ singlet. In terms of spin operators,
\[
    \mu_d(e_x) = -2\big(S^1_xS^1_{x+1}+S^2_xS^2_{x+1}+\Delta\,S^3_xS^3_{x+1}\big)
    + \tfrac{q-q^{-1}}{2}\big(S^3_x - S^3_{x+1}\big) + \tfrac{\Delta}{2}\,\ds\,,
\]
in which the middle term telescopes on summation, so that
\begin{equation}\label{eq:xxz_kink}
    \mu_d\left(\sum_{x=1}^{L-1}e_x\right)
    = -2\sum_{x=1}^{L-1}\big(S^1_xS^1_{x+1}+S^2_xS^2_{x+1}+\Delta\,S^3_xS^3_{x+1}\big)
    + \tfrac{\sqrt{d^2-4}}{2}\big(S^3_1-S^3_L\big) + \tfrac{\Delta}{2}(L-1)\,\ds\,.
\end{equation}
A direct computation verifies \eqref{eq:TL_relations} at $\lambda = d$. Multiplying by a factor of $\frac{1}{d}$ gives the ferromagnetic $XXZ$ chain at anisotropy $\Delta = d/2$ carrying the kink boundary fields as presented by \cite{nachtergaele2004ferromagnetic}, who show the spectral gap of the model on an open chain of length $L$ is $1-\frac{2}{d}\,\cos\left(\tfrac{\pi}{L}\right)$; see also \cite{koma1997spectral}. The boundary fields are what make the Hamiltonian $U_q(\mathfrak{sl}_2)$ invariant. The resulting decomposition is what we use in Section~\ref{sec:foel_all} to argue $\mu_d$ is faithful. 

It is not a coincidence that this is the spectral gap we obtain as well for the singlet model. All of these presentations of the singlet and $XXZ$ chains produce the same abstract Temperley--Lieb element $\sum_x e_x$ at the same loop parameter $d$. On the open chain this is already enough to force the spectra to agree (once we show faithfulness of the higher spin representations $\rho_d$ in Section~\ref{sec:reptheory}).

\begin{prop}[Isospectrality for open chains]\label{prop:isospectral}
Let $d \geq 2$ be an integer, $h = \sum_{x=1}^{L-1}e_x \in TL_L(d)$, and $\rho$ and $\mu$ be finite dimensional faithful $TL_L(d)$-representations. Then
\[
    \sigma\big(\rho(h)\big) = \sigma\big(\mu(h)\big)\,,
\]
that is, the two operators have the same set of eigenvalues.
\end{prop}
\begin{proof}
This is by a direct application of Corollary~\ref{cor:isospectral_general}.
\end{proof}

In particular the $SU(2)$ singlet model for spin $s$, the alternating $SU(d)$ chain, and the $O(d)$ singlet chain are mutually isospectral, and all three are isospectral to $\tfrac{1}{d}$ times the reference chain \eqref{eq:xxz_kink}. The multiplicities, by contrast, are highly realization-dependent, as we discuss in Section~\ref{sec:reptheory}.

Proposition~\ref{prop:isospectral} does not extend to $H^{\text{PBC}}_L$. The wrap-around term $P^{(0)}_{L,1}$ in \eqref{eq:hamiltonian_pbc} is not a product of the generators $E_1,\dots,E_{L-1}$, so the relevant algebra is the periodic Temperley--Lieb algebra, and the proof fails. The periodic spectra are not contained in that of any single $XXZ$ chain: as shown in \cite{aufgebauer2010quantum}, the Hilbert space decomposes instead into representations each isomorphic to a sector of an $XXZ$ chain with twisted boundary conditions, so that recovering the spectrum of $H^{\text{PBC}}_L$ requires a whole family of reference chains.

\subsection{The Infinite Volume Setting}
\label{sec:cstar}

We recall operator-algebraic vocabulary needed to formulate the infinite volume results in Section~\ref{sec:gap}; see \cite{nachtergaele2019quasilocality} for a complete treatment. Let $\cA$ denote the quasi-local algebra on $\zz$, that is, the norm closure of the local algebra $\bigcup_X \cA_X$, whose union is over finite subsets of the lattice and each $\cA_X = M_d(\cc)^{\otimes \abs{X}}$. An interaction $\Phi$ is a map between finite subsets of the lattice and observables supported on that subset, $X \mapsto \Phi(X) \in \cA_X$, such that each $\Phi(X)$ is self-adjoint. Hamiltonians on a finite region $\Lambda$ are given by sums of interaction terms over subsets of $\Lambda$: $H_\Lambda^\Phi = \sum_{X \subseteq \Lambda} \Phi(X)$. We now fix notation for the rest of the paper. The open chain interaction is
\begin{equation}\label{eq:interaction_obc}
    \Phi(X) = \begin{cases} P^{(0)}_{x,x+1} & X = \{x,x+1\}\,,\\ 0 & \text{otherwise,}\end{cases}
\end{equation}
so that $H_{[1,L]}^\Phi = H_L$. For a sequence of volumes $\Lambda_n = [-L_n,L_n]$ with $L_n \to \infty$, the periodic interactions are
\begin{equation}\label{eq:interaction_pbc}
    \Phi_n(X) = \begin{cases} P^{(0)}_{x,x+1} & X = \{x,x+1\} \subset \Lambda_n \text{ and } \ x=-L_n,\dots,L_n-1\,,\\ P^{(0)}_{L_n,-L_n} & X = \{L_n,-L_n\}\,,\\ 0 & \text{otherwise,}\end{cases}
\end{equation}
so that $H^{\Phi_n}_{\Lambda_n} = H_{2L_n+1}^{\text{PBC}}$. We denote $N_n = 2L_n+1$ for the length of $\Lambda_n$. Note $\Phi_n$ has terms only supported in $\Lambda_n$, while $\Phi$ has nonzero terms on all bonds of $\zz$. We let $\tau_t$ denote the infinite volume dynamics coming from $\Phi$, and $i\delta$ its generator. We similarly define $\tau_t^{\Phi_n}$ and $i\delta^{\Phi_n}$.

The fully polarized vectors define states $\omega_\Lambda(A) = \bra{s\cdots s}A\ket{s\cdots s}$ on $\cA_\Lambda$, and the family $\{\omega_\Lambda\}_{\Lambda}$ (indexed by finite subsets of the lattice $\Lambda$) is consistent: $\omega_{\Lambda'}\big|_{\cA_\Lambda} = \omega_\Lambda$ whenever $\Lambda \subseteq \Lambda'$. Therefore it has a unique weak-$*$ limit
\begin{equation}\label{eq:omega}
    \omega : \cA \to \cc\,, \qquad \omega = \lim_{\Lambda \uparrow \zz}\omega_\Lambda\,.
\end{equation}
$\omega$ is a ground state of $(\cA, \tau_t)$ and of $(\cA, \tau_t^{\Phi_n})$ for every $n$. In Section~\ref{sec:other_ground_states} we give another family of finite volume ground states $\{\bar\omega_\Lambda\}$ which is consistent when indexed by \emph{intervals} $\Lambda$ rather than arbitrary finite subsets. These will be the maximally mixed ground states of the model, and they again have a unique weak-$*$ limiting infinite volume ground state. We write $(\cH_\omega, \pi, \Omega)$ for the GNS triple of $\omega$ with respect to $(\cA, \tau_t)$ and $H$ for the associated GNS Hamiltonian, i.e. the unique positive self-adjoint operator on $\cH_\omega$ with $H\Omega = 0$ and $e^{itH}\pi(A)e^{-itH} = \pi(\tau_t(A))$.

\subsection{Main Results}
\label{sec:results}

Our first result concerns the ordering of the energy levels across the summands $W_{L,k}\otimes U_k$ into which $\cH_L$ decomposes from the joint action of $TL_L(d)$ and its commutant $\rmEnd_{TL_L(d)}(\cH_L)$; the precise statement, and the use of link diagrams to equivalently label the sectors as $k$ does, are recalled in Section~\ref{sec:reptheory}. On the $k$-th summand $H_L$ acts as $\rho_{d,k}(H_L)\otimes\ds_{U_k}$ (see Proposition~\ref{prop:decomposition}), so its lowest energy there, which we write $\cE(L,k)$, is determined by the standard module alone, $U_k$ contributing only the multiplicity. We prove the analogue for the singlet model of the ferromagnetic ordering of energy levels established for the spin-$\tfrac12$ $XXZ$ chain in \cite{nachtergaele2004ferromagnetic}, with $k$ in the role of decreasing total spin.

\begin{thm}[Ferromagnetic ordering of energy levels]\label{thm:foel}
For every $L \geq 2$,
\begin{equation}\label{eq:foel}
    \cE(L,0) < \cE(L,1) < \cdots < \cE(L,\lfloor L/2\rfloor)\,.
\end{equation}
\end{thm}

Since $\cE(L,0) = 0$ by frustration-freeness, Theorem~\ref{thm:foel} identifies the first excited energy of $H_L$ as $\cE(L,1)$, and the sector $k=1$ is small enough to be diagonalized in closed form. This yields the exact finite volume spectral gaps.

\begin{cor}\label{cor:fv_gap}
For every $L \geq 2$, the spectral gap of $H_L$ is
\begin{equation}\label{eq:fv_gap}
    \gamma_L = 1 - \frac{2}{d}\cos\!\left(\frac{\pi}{L}\right)\,.
\end{equation}
In particular the infimum of the finite volume gaps is $\inf_L \gamma_L = 1 - \tfrac{2}{d} > 0$.
\end{cor}

Theorem~\ref{thm:foel} and Corollary~\ref{cor:fv_gap} are proved in Section~\ref{sec:foel}. Our second main result refers to the GNS gap, and shows that the value $1-\tfrac{2}{d}$ appearing in Corollary~\ref{cor:fv_gap} is the exact gap of the GNS Hamiltonian $H$.

\begin{thm}[GNS gap]\label{thm:gns_gap}
Let $H$ be the GNS Hamiltonian of the fully polarized ground state $\omega$ defined above. Then $H$ has spectral gap exactly $1 - \tfrac{2}{d}$.
\end{thm}

The upper bound is obtained in Section~\ref{sec:upper} by exhibiting an explicit Weyl sequence built from single- and two-site excitations of the fully polarized state, which shows that the essential spectrum of $H$ contains the whole band $[1-\tfrac2d, 1+\tfrac2d]$. The matching lower bound is obtained in Section~\ref{sec:lower} from Knabe's finite size criterion applied at two bonds, together with an exact diagonalization of the three-site open chain. Related to the above theorem is the result of Proposition~\ref{prop:tracial_gapped}, which shows another infinite volume ground state of $(\cA, \tau_t)$ has GNS gap bounded below by $1-2/d$, namely the weak-$*$ limit obtained from the maximally mixed ground state on finite open chains.

Our third main result is an operator identity for the ground state projectors of the open chains. Let $P_L$ denote the orthogonal projection onto $\ker H_L$ and let $D_L = \Tr P_L$ be the ground state degeneracy. The operator identity is derived in standard texts on the Jones--Wenzl relations \cite{wenzl1987sequences, lickorish1997introduction, kauffman1994temperley}. The identification of the assumptions used in these derivations with the required properties of ground state projectors is new. The formula for the degeneracies was derived in \cite{aufgebauer2010quantum} by a different argument. 

\begin{thm}[Jones--Wenzl recursion]\label{thm:jw}
Define $P_1 = \ds$, $\Omega_0 = 1, \Omega_1 = d$ and $\Omega_L = d\,\Omega_{L-1} - \Omega_{L-2}$ for all $L \geq 2$. Then for all $L \geq 2$,
\begin{equation}\label{eq:jw_recursion}
    P_L = P_{L-1} - \frac{\Omega_{L-2}}{\Omega_{L-1}}\,P_{L-1}E_{L-1}P_{L-1}\,, \qquad E_{L-1} = d\,P^{(0)}_{L-1,L}\,.
\end{equation}
Consequently $D_L = \Omega_L$ for all $L \geq 2$, and writing $q > 1$ for the root of $q + q^{-1} = d$,
\begin{equation}\label{eq:degeneracy}
    D_L = [L+1]_q = \frac{q^{L+1}-q^{-(L+1)}}{q - q^{-1}}\,.
\end{equation}
Here $P_{L-1}$ is understood as the ground state projector on the first $L-1$ sites of $[1,L]$, that is the projector onto $\ker\, (H_{L-1} \otimes \ds)$.
\end{thm}

This theorem also shows that the ground state degeneracy grows exponentially for our case of $s \geq 1$, in contrast with the linear growth $2sL+1$ for ferromagnetic chains discussed in Section~\ref{sec:instability}, at whose phase boundary the singlet model sits (see Figure~\ref{fig:Spin1 Diagram}). Theorem~\ref{thm:jw} is proved in Section~\ref{sec:projectors}, where the recursion is also used to characterize $P_L$ uniquely among operators satisfying an algebraic condition valid for any frustration free model. Section~\ref{sec:discussion} discusses two consequences of this exponential degeneracy: that the model evades the conclusion of Goldstone's theorem, breaking a continuous symmetry while remaining gapped, and that the gap of Theorem~\ref{thm:gns_gap} is unstable.

\section{Ferromagnetic Ordering of Energy Levels}
\label{sec:foel}

In this section we prove Theorem~\ref{thm:foel}. It can be shown for all spins using the results of \cite{nachtergaele2004ferromagnetic}, and the spectral equivalence between the $XXZ$ chain and the singlet model. This is argued in Section~\ref{sec:foel_all}. It can be shown directly as well only for half-odd integer spins using similar arguments to \cite{nachtergaele2004ferromagnetic}, with the restriction on $s$ due to the parity $(-1)^{2s}$ of the singlet vector. This is done in the remaining sections.

\subsection{Decomposition of the Chain into Symmetry Sectors}
\label{sec:reptheory}

The symmetry algebra of the singlet model is the commutant $\rmEnd_{TL_L(d)}(\cH_L)\subset\rmEnd(\cH_L)$ of the action of $TL_L(d)$ under $\rho_d$. The double centralizer theorem (Theorem~\ref{thm:double_centralizer}) gives that this algebra is semisimple, and that $\cH_L$ decomposes into sectors indexed by the irreducible representations of either algebra. We now describe these sectors concretely, following \cite{read2007enlarged}; see also \cite{kulish2003spin} and \cite{kulish2008quantum} for the description derived using the $R$-matrix formalism. We introduce a convenient basis for the simple modules of $TL_L(d)$.

\begin{definition}[Link patterns]\label{def:link}
A \emph{link pattern} on $L$ sites with $k$ arcs is a partition of a subset of $\{1,\dots,L\}$ of size $2k$ into $k$ pairs, drawn as arcs above the line of sites, subject to the two conditions that the arcs be non-crossing and that no arc span an unpaired site. Unpaired sites are called \emph{dots}, or through-lines; there are $L-2k$ of them. We write $B_{L,k}$ for the set of such patterns and $W_{L,k}$ for the complex vector space they span.
\end{definition}

Equivalently, and as in \cite{read2007enlarged}, a link pattern is a configuration of nested parentheses and dots obeying the usual typographical rules: each "(" is matched by a single ")", and no dot appears between a matched pair. The generators of $TL_L(d)$ act on $W_{L,k}$ diagrammatically: $e_x$ acts on the pattern $\pi$ by adjoining a cup-cap at the sites $x,x+1$ and reading off the resulting pattern, giving $d\pi$ if $x$ and $x+1$ are already joined by an arc of $\pi$, giving $0$ if both are dots, and giving a new pattern with the same number of arcs otherwise. This makes $W_{L,k}$ a $TL_L(d)$ module, the \emph{standard module}; for $d \geq 2$ these are irreducible and exhaust the irreducible representations as $k$ ranges over $0,1,\dots,\lfloor L/2\rfloor$. Figure~\ref{fig:linkpatterns} illustrates each of these cases; note in particular that the number of arcs is never increased, which is what makes $W_{L,k}$ invariant.

\begin{figure}[htbp]
\centering
\begin{tikzpicture}[scale=0.56]

\begin{scope}[yshift=0cm]
  \node[anchor=east,font=\small] at (0.0,0.5) {(a)};
  \lpsites{6}\lplabels{6}
  \lparc{3}{4}\lpthru{1}\lpthru{2}\lpthru{5}\lpthru{6}
  \node[font=\small] at (3.5,1.6) {$\pi$};
  \node[font=\small] at (7.9,0.5) {$\xrightarrow{\;e_3\;}$};
  \begin{scope}[xshift=9.4cm]
    \lpsites{6}\lplabels{6}
    \lparc{3}{4}\lpthru{1}\lpthru{2}\lpthru{5}\lpthru{6}
    \node[font=\small] at (3.5,1.6) {$d\,\pi$};
  \end{scope}
\end{scope}

\begin{scope}[yshift=-2.8cm]
  \node[anchor=east,font=\small] at (0.0,0.5) {(b)};
  \lpsites{6}\lplabels{6}
  \lparc{3}{4}\lpthru{1}\lpthru{2}\lpthru{5}\lpthru{6}
  \node[font=\small] at (3.5,1.6) {$\pi$};
  \node[font=\small] at (7.9,0.5) {$\xrightarrow{\;e_1\;}$};
  \node[font=\small] at (10.0,0.5) {$0$};
\end{scope}

\begin{scope}[yshift=-5.6cm]
  \node[anchor=east,font=\small] at (0.0,0.5) {(c)};
  \lpsites{6}\lplabels{6}
  \lparc{3}{4}\lpthru{1}\lpthru{2}\lpthru{5}\lpthru{6}
  \node[font=\small] at (3.5,1.6) {$\pi$};
  \node[font=\small] at (7.9,0.5) {$\xrightarrow{\;e_2\;}$};
  \begin{scope}[xshift=9.4cm]
    \lpsites{6}\lplabels{6}
    \lparc{2}{3}\lpthru{1}\lpthru{4}\lpthru{5}\lpthru{6}
    \node[font=\small] at (3.5,1.6) {$\pi'$};
  \end{scope}
\end{scope}

\begin{scope}[yshift=-8.6cm]
  \node[anchor=east,font=\small] at (0.0,0.5) {(d)};
  \lpsites{6}\lplabels{6}
  \lparc{2}{3}\lparc{4}{5}\lpthru{1}\lpthru{6}
  \node[font=\small] at (3.5,1.9) {$\tilde\pi$};
  \node[font=\small] at (7.9,0.5) {$\xrightarrow{\;e_3\;}$};
  \begin{scope}[xshift=9.4cm]
    \lpsites{6}\lplabels{6}
    \lparc{3}{4}\lparc{2}{5}\lpthru{1}\lpthru{6}
    \node[font=\small] at (3.5,1.9) {$\tilde\pi'$};
  \end{scope}
\end{scope}

\end{tikzpicture}
\caption{The action of the generators on link patterns, drawn on $L=6$ sites with arcs above the line and through-lines as vertical segments.}
\label{fig:linkpatterns}
\end{figure}
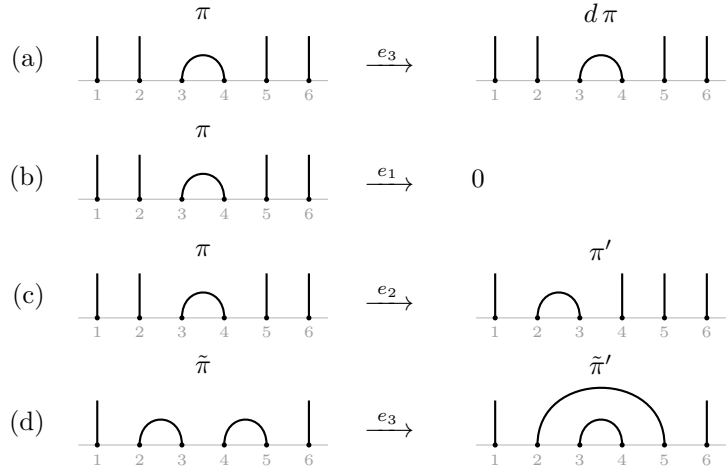

The realization of $W_{L,k}$ inside $\cH_L$ is obtained by placing a spin singlet $\ket{\psi}$ on each arc and a suitable vector on the dots. Let $\pi$ be a link pattern with arcs $(i_1<i_1'),\dots,(i_k<i_k')$, and label the remaining sites (the dots) as $j_1,...,j_{L-2k}$. Define
\begin{equation}\label{eq:jwp}
    U_k = \bigcap_{x=1}^{L-2k-1}\ker P^{(0)}_{j_x,j_{x+1}} \subset V_s^{\otimes (L-2k)}\,.
\end{equation}
Under the order-preserving identification of the dot sites with $\{1,\ldots,L-2k\}$, this definition depends only on the number $L-2k$ of dots, and not on their positions in a particular link pattern. This space is nonzero because the fully polarized state $\ket{s\cdots s}$ lives inside it. By Theorem~\ref{thm:jw}, this space can also be characterized as $\Im \rho_d(p_{L-2k})$ where $p_{L-2k} \in TL_{L-2k}(d)$ is the Jones--Wenzl projector. As a consequence, the dimension of $U_k$ is $[L-2k+1]_q$. That this projector is an element of the Temperley--Lieb algebra follows from standard arguments of its construction. Then to such a pattern $\pi$, and vector $\ket{V} \in U_k$ on the dots we associate
\begin{equation}\label{eq:pattern_vector}
    \Psi_{\pi, V} = \ket{V}_{\mathrm{dots}} \otimes \bigotimes_{r=1}^{k}\ket{\psi}_{i_r i_r'} \in \cH_L\,.
\end{equation}
The condition $\ket{V} \in U_k$ says precisely that applying an interaction term to any two consecutive dots gives $0$. The generators $\{E_x\}_{x=1}^{L-1}$ of $TL_L(d)$ under the representation $\rho_d$ act on these states according to the following lemma. This is the diagrammatic action following Definition~\ref{def:link}.

\begin{lem}[Local action of the interaction]\label{lem:local_action}
Let $\ket{v} \in V_s$ denote the tensor factor of a term in $\ket{V}$ at the indicated dot. Then the interaction $P^{(0)}_{x,x+1}$ acts on $\Psi_{\pi,V}$ as follows.
\begin{enumerate}
    \item[\textup{(C0)}] If $x,x+1$ are joined by an arc of $\pi$, then $P^{(0)}_{x,x+1}\Psi_{\pi,V} = \Psi_{\pi,V}$.
    \item[\textup{(C1)}] If $x$ and $x+1$ are both dots, then $P^{(0)}_{x,x+1}\Psi_{\pi,V} = 0$.
    \item[\textup{(C2)}] If $x$ is a dot and $x+1$ belongs to the arc $(x+1,l)$ with $l > x+1$, then
    \[P^{(0)}_{x,x+1}\Big(\ket{v}_x \otimes \ket{\psi}_{x+1,l}\Big) = \frac{(-1)^{2s}}{d}\,\ket{\psi}_{x,x+1}\otimes \ket{v}_l\,.\]
    \item[\textup{(C3)}] If $x$ belongs to the arc $(l,x)$ with $l<x$ and $x+1$ is a dot, then
    \[P^{(0)}_{x,x+1}\Big(\ket{\psi}_{l,x}\otimes\ket{v}_{x+1}\Big) = \frac{(-1)^{2s}}{d}\,\ket{v}_l \otimes \ket{\psi}_{x,x+1}\,.\]
    \item[\textup{(C4)}] If $x$ and $x+1$ belong to two distinct arcs $(l,x)$ and $(x+1,m)$, then
    \[P^{(0)}_{x,x+1}\Big(\ket{\psi}_{l,x}\otimes\ket{\psi}_{x+1,m}\Big) = \frac{(-1)^{2s}}{d} \,\ket{\psi}_{l,m}\otimes\ket{\psi}_{x,x+1}\,,\]
    and similarly in the other two cases where $x$ and $x+1$ are on distinct arcs: $x<x+1<m<l$ and $m<l<x<x+1$.
\end{enumerate}
\end{lem}
\begin{proof}
(C0) is idempotency and (C1) is the defining property \eqref{eq:jwp} of $U_k$. For the remaining cases everything follows from a single contraction. Writing $P^{(0)}_{x,x+1} = \tfrac1d\sum_{\ga,\eta}(-1)^{2s-\ga-\eta}\ket{\ga,-\ga}_{x,x+1}\bra{\eta,-\eta}_{x,x+1}$ and $\ket{v} = \sum_\bet v_\bet\ket{\bet}$, case (C2) (and similarly (C3)) reads
\begin{align*}
    P^{(0)}_{x,x+1}\Big(\ket{v}_x \otimes \ket{\psi}_{x+1,l}\Big)
    &= \frac{1}{d}\sum_{\bet,\al,\ga,\eta} v_\bet(-1)^{s-\al}(-1)^{2s-\ga-\eta}\,\delta_{\eta,\bet}\,\delta_{-\eta,\al}\,\ket{\ga}_x\ket{-\ga}_{x+1}\ket{-\al}_l\\
    &= \frac{1}{d}\sum_{\bet,\ga} v_\bet(-1)^{s+\bet}(-1)^{2s-\ga-\bet}\,\ket{\ga}_x\ket{-\ga}_{x+1}\ket{\bet}_l\\
    &= \frac{(-1)^{2s}}{d}\,\Big[\sum_\ga (-1)^{s-\ga}\ket{\ga,-\ga}_{x,x+1}\Big]\otimes \sum_\bet v_\bet \ket{\bet}_l\\
    &= \frac{(-1)^{2s}}{d}\,\ket{\psi}_{x,x+1} \otimes \ket{v}_l\,.
\end{align*}

The sub-case where $l < x < x+1 < m$ in Case (C4) (and similarly for the other sub-cases of (C4)) reads
\begin{align*}
    P^{(0)}_{x,x+1}&\Big(\ket{\psi}_{l,x}\otimes\ket{\psi}_{x+1,m}\Big)\\
    &= \frac{1}{d}\sum_{\al,\bet,\ga,\eta}(-1)^{s-\al}(-1)^{s-\bet}(-1)^{2s-\ga-\eta}\,\delta_{\eta,-\al}\,\delta_{-\eta,\bet}\;\ket{\al}_l\ket{\ga}_x\ket{-\ga}_{x+1}\ket{-\bet}_m\\
    &= \frac{1}{d}\sum_{\al,\ga}(-1)^{s-\al}(-1)^{s-\al}(-1)^{2s-\ga+\al}\;\ket{\al}_l\ket{\ga}_x\ket{-\ga}_{x+1}\ket{-\al}_m\\
    &= \frac{(-1)^{2s}}{d}\Big[\sum_{\al}(-1)^{s-\al}\ket{\al}_l\ket{-\al}_m\Big]\otimes\Big[\sum_{\ga}(-1)^{s-\ga}\ket{\ga,-\ga}_{x,x+1}\Big]\\
    &= \frac{(-1)^{2s}}{d}\,\ket{\psi}_{l,m}\otimes\ket{\psi}_{x,x+1}\,.
\end{align*}
\end{proof}

This allows us to conclude the argument there is a copy of $W_{L,k}$ in $\cH_L$. Fix $k \leq \lfloor L/2\rfloor$ and a nonzero $\ket{V} \in U_k$. Set $\Sigma(\pi) = \sum_r i_r$, the sum of the left endpoints of the arcs of $\pi$, and $\varepsilon(\pi) = (-1)^{2s\Sigma(\pi)}$. Every arc has odd length: if $(i,j)$ is an arc, the sites $i+1,...,j-1$ are not dots, since no arc spans a dot, and being non-crossing they are matched among themselves, so $j-i$ is odd.

The rules for the arcs say they span a subset of the chain of odd length. Suppose now that one of the cases (C2)--(C4) of Lemma~\ref{lem:local_action} applies, so that $\pi' \coloneq e_x\pi$ is a pattern different from $\pi$, and put $\sigma = \Sigma(\pi')-\Sigma(\pi)$. Only the arcs meeting $\{x,x+1\}$ change. In the case (C2) where $(x+1,l) \rightsquigarrow (x,x+1)$, we have $\sigma = -1$. In the case (C3) where $(l,x) \rightsquigarrow (x,x+1)$, we have $\sigma = x-l$ is odd as $(l,x)$ was an arc. In the first subcase of (C4) where $(l,x),(x+1,m) \rightsquigarrow (l,m),(x,x+1)$, we have $\sigma = -1$. Similar arguments for the other two subcases of (C4) give $\sigma$ is odd in every case and so
\[
    \varepsilon(\pi') = (-1)^{2s\Sigma(\pi)}\big[(-1)^{2s}\big]^{\sigma} = (-1)^{2s}\,\varepsilon(\pi)\,.
\]
Define the map $\Theta(\pi) = \varepsilon(\pi)\Psi_{\pi,V}$. Therefore, in cases (C2)--(C4) we have
\begin{align*}
E_x \Theta(\pi) &= (-1)^{2s}\varepsilon(\pi)\Psi_{\pi',V} \\
                &= (-1)^{2s}\varepsilon(\pi)\varepsilon(\pi')\Theta(\pi') \\
                &= (-1)^{2s}\varepsilon(\pi)\varepsilon(\pi')\Theta(e_x\pi) \\
                &= \left((-1)^{2s}\right)^2\varepsilon(\pi)^2\Theta(e_x\pi) = \Theta(e_x\pi)
\end{align*}
In case (C0) both sides acquire the factor $d$ and in case (C1) both vanish. Hence $\Theta$ intertwines the action of $TL_L(d)$ on $W_{L,k}$ with its action on $\cH_L$. $\Theta$ is nonzero since $\Psi_{\pi,V} \neq 0$ for $\ket{V} \neq 0$. By simplicity of $W_{L,k}$, $\Theta$ is injective. Its image is therefore a copy of $W_{L,k}$ inside $\cH_L$, and in particular the vectors $\{\Psi_{\pi,V}:\pi\in B_{L,k}\}$ are linearly independent for each nonzero $\ket{V}\in U_k$. By Proposition~\ref{prop:faithful_iff}, $\rho_d$ is faithful. Letting $\ket{V}$ vary gives $U_k$ is the multiplicity space for $W_{L,k}$. We now have by the double centralizer theorem, Theorem~\ref{thm:double_centralizer}, the following result.
\begin{prop}[Sector decomposition]\label{prop:decomposition}
As a module of $TL_L(d)$ under $\rho_d$, 
\begin{equation}\label{eq:decomposition}
    \cH_L \cong \bigoplus_{k=0}^{\lfloor L/2\rfloor} W_{L,k}\otimes U_k\,.
\end{equation}
For a fixed $\ket{V} \in U_k$, letting the pattern vary in \eqref{eq:pattern_vector} gives a basis of a copy of $W_{L,k}$; for a fixed pattern $\pi$, letting $\ket{V}$ vary gives a copy of $U_k$. Moreover $H_L$ acts on the $k$-th summand as $\rho_{d,k}(H_L)\otimes \ds_{U_k}$, where $\rho_{d,k}$ denotes the restriction of $\rho_d$ to the $k$-th block, that is the action of $TL_L(d)$ on the standard module $W_{L,k}$.
\end{prop}

The last statement is the reason the analysis below never needs to know what the dots are: a Temperley--Lieb generator transports the state on the sequence of dots without altering it. Equivalently, the vectors $\{\Psi_{\pi,V}\}_{\pi \in B_{L,k}}$ span an $H_L$ invariant subspace for each fixed $\ket{V}$, on which $H_L$ is represented by a matrix independent of the choice of $\ket{V}$.

It is worth recording what \eqref{eq:decomposition} specializes to at $s=\tfrac12$, since the general statement reduces to a classical one. There the singlet is the antisymmetric vector, so $P^{(0)} = \tfrac12(\ds-T)$ with $T$ the transposition of the two factors, and the generators \eqref{eq:Ex} are $E_x = \ds - T_{x,x+1}$. As $T_{x,x+1} = \ds - E_x$ and the adjacent transpositions generate $S_L$, the unital algebra generated by the $E_x$ is the image of $\cc[S_L]$ acting on $\cH_L$ by permutation of the tensor factors, and \eqref{eq:decomposition} becomes
\begin{equation}\label{eq:schur_weyl}
    (\cc^2)^{\otimes L} \;\cong\; \bigoplus_{k=0}^{\lfloor L/2\rfloor}
    S^{(L-k,k)}\otimes V_{(L-k,k)}\,,
\end{equation}
the $SU(2)$ case of Schur--Weyl duality, with $S^\lambda$ the Specht module and $V_\lambda$ the irreducible $SU(2)$-module attached to a partition $\lambda$ of $L$, drawn as a Ferrers-Young diagram of $L$ boxes. Only diagrams with at most two rows occur, since three vectors in $\cc^2$ cannot be antisymmetrized; this is the same restriction as $k \leq \lfloor L/2\rfloor$. The dictionary is that $k$ is the length of the second row and $L-2k$ the difference of the two row lengths, so that $W_{L,k} = S^{(L-k,k)}$ has dimension $\binom{L}{k}-\binom{L}{k-1}$, the number of link patterns in $B_{L,k}$, while $U_k = V_{(L-k,k)}$ is the multiplet of spin $\tfrac12(L-2k)$, of dimension $L-2k+1 = [L-2k+1]_q$ at $q=1$. Consistently, $q=1$ makes the Jones--Wenzl idempotent $p_{L-2k}$ the full symmetrizer, so that $U_k = \Im\rho_2(p_{L-2k}) = \mathrm{Sym}^{L-2k}(\cc^2)$: in \eqref{eq:pattern_vector} the arcs carry singlets and contribute no spin, while the through-lines are symmetrized into the maximal total spin available to them. Note then the dimensions of the standard modules $W_{L,k}$ do not depend on $s$, only the dimension of $U_k$ does, in which it grows exponentially for $s \geq 1$, i.e., for $q > 1$.

\subsection{FOEL for all Spins}
\label{sec:foel_all}

The proof of Theorem~\ref{thm:foel} given in Sections~\ref{sec:monotonicity}--\ref{sec:upgrading} is diagrammatic, and it is sensitive to the parity of $2s$ through the reversal rule \eqref{eq:reversal}: Part (i) of Proposition~\ref{prop:comparison_lemma_hypothesis} needs $(-1)^{2s} = -1$ to obtain non-positive off diagonal entries as required by Perron-Frobenius. We now give a proof for all $s$ by a different route, which uses no diagrammatics at all. The observation behind it is that the quantities $\cE(L,k)$ of Definition~\ref{def:energies} are attached to the algebra $TL_L(d)$ and not to any realization of it.

\begin{definition}\label{def:energies}
For $0 \leq k \leq \lfloor L/2\rfloor$ let $A_{L,k}$ denote the matrix of $H_L$ in the basis $\{\Psi_{\pi,V}\}_{\pi \in B_{L,k}}$ of one copy of $W_{L,k}$, defined by
\begin{equation}\label{eq:ALk}
    H_L\,\Psi_{\pi,V} = \sum_{\pi' \in B_{L,k}} A_{L,k}^{\pi'\pi}\,\Psi_{\pi',V}\,,
\end{equation}
and set
\begin{equation}\label{eq:Elk}
    \cE(L,k) \coloneq \min \sigma\big(H_L\big|_{W_{L,k}\otimes U_k}\big) = \min\sigma(A_{L,k})\,,
\end{equation}
with the convention $\cE(L,k) \coloneq +\infty$ whenever $W_{L,k}\otimes U_k = 0$, that is whenever $k < 0$ or $k > \lfloor L/2\rfloor$.
\end{definition}

As mentioned previously, the matrix $A_{L,k}$ does not depend on the choice of $\ket{V}$, as $U_k$ only serves as a multiplicity space. Note also $W_{L,0}$ is one dimensional, spanned by the pattern with no arcs, and every generator annihilates it; hence $\cE(L,0) = 0$ and the ground state space of $H_L$ is exactly the $k=0$ sector $U_0 = \Im \rho_d(p_L)$. This is the statement in Theorem~\ref{thm:jw}.

By Proposition~\ref{prop:decomposition} the Hamiltonian acts on the $k$-th summand as $\rho_{d,k}(H_L)\otimes\ds_{U_k}$, so
\begin{equation}\label{eq:E_intrinsic}
    \cE(L,k) \;=\; \frac1d\,\min\sigma\Big(\rho_{d,k}\big(\textstyle\sum_{x=1}^{L-1}e_x\big)\Big)\,.
\end{equation}
The standard module $W_{L,k}$ is determined up to isomorphism by $L$, $k$ and the loop parameter $d$ alone. By Lemma~\ref{lem:iso_spectrum} the right hand side of \eqref{eq:E_intrinsic} is unchanged if $\rho_{d,k}$ is replaced by the action of $TL_L(d)$ on any other copy of $W_{L,k}$. Thus $\cE(L,k)$ can be computed in whichever realization of $TL_L(d)$ is most convenient, and the computation for the $XXZ$ chain from the realization $\mu_d$ of Section~\ref{sec:model} is available at every loop parameter $d$ of interest.

This relies on the fact $\mu_d$ is faithful, in other words, by Proposition~\ref{prop:faithful_iff}, there is at least one occurrence of $W_{L,k}$ for each possible $k$ in the Hilbert space $\big(\cc^2\big)^{\otimes L}$. To see this, we can construct the modules $W_{L,k}$ in the same way as we did for the singlet model, where arcs are no longer $SU(2)$ spin singlets, but rather the $U_q(\mathfrak{sl}_2)$ singlet $\ket{\psi_q}$. The dots then correspond to a $U_q(\mathfrak{sl}_2)$ multiplet of spin $\tfrac{L}{2}-k$, which is of dimension $L-2k+1$. Then the $XXZ$ chain acts on these states just as we had for the singlet model using Lemma~\ref{lem:local_action}. The double centralizer decomposition then reads
\begin{equation}\label{eq:xxz_decomposition}
    \big(\cc^2\big)^{\otimes L} \;\cong\; \bigoplus_{k=0}^{\lfloor L/2\rfloor}
    W_{L,k}\otimes U_k^{XXZ}\,, \qquad \dim U_k^{XXZ} = L-2k+1\,,
\end{equation}
with $W_{L,k}$ the same standard modules as in Proposition~\ref{prop:decomposition} and $U_k^{XXZ}$ the irreducible $U_q(\mathfrak{sl}_2)$ module of spin $J = L/2-k$: decreasing total spin corresponds to increasing the number of arcs.

Write $H^{XXZ}_L$ for the $XXZ$ chain as written in \cite{nachtergaele2004ferromagnetic}. Recall $\mu_d\left(\sum_{x=1}^{L-1} e_x\right) = dH^{XXZ}_L$. Let $E^{XXZ}(L,k)$ denote the lowest energy of $H^{XXZ}_L$ in the spin $J = L/2-k$ sector. Evaluating \eqref{eq:E_intrinsic} in the realization $\mu_d$ then gives $\cE(L,k) = E^{XXZ}(L,k)$. The ordering of the $E^{XXZ}(L,k)$ established in \cite{nachtergaele2004ferromagnetic} reads as
\[
E^{XXZ}(L,0) < E^{XXZ}(L,1) < \cdots < E^{XXZ}(L,\lfloor L/2\rfloor)\,,
\]
and therefore passes to the $\cE(L,k)$ as in Theorem~\ref{thm:foel} under the exact relationship between the two.

We remark that the sector label in Theorem~\ref{thm:foel} cannot be replaced by the $SU(2)$ total spin once $s \geq 1$. Writing $\widetilde{\cE}(L,J)$ for the lowest energy of $H_L$ in the total spin $J$ subspace of $\cH_L$, each $\widetilde{\cE}(L,J)$ is the minimum of $\cE(L,k)$ over those $k$ for which $J$ occurs in $U_k$, and for $s \geq 1$ the spaces $U_k$ contain many $SU(2)$ representations. Already at $s=1$ and $L=2$ the ground state space is $U_0 \cong V_2\oplus V_1$, so that $\widetilde{\cE}(2,2) = \widetilde{\cE}(2,1) = 0$. The ordering is therefore genuinely a statement about the sectors labeled by the number of arcs, and has no counterpart in terms of $SU(2)$ alone.

\subsection{Monotonicity in the Volume}
\label{sec:monotonicity}

We now begin to review how to generalize the strategies of \cite{nachtergaele2004ferromagnetic} to directly show FOEL for the half-odd integer spin singlet model. The strategy of \cite{nachtergaele2004ferromagnetic} is to deduce the ordering from strict monotonicity of $\cE(\cdot,k)$ in the volume. That monotonicity follows in turn from a comparison lemma for matrices with non-positive off-diagonal entries, which we take from \cite[Lemma 7.3]{nachtergaele2004ferromagnetic} and the remarks after it.

\begin{lem}\label{lem:ns73}
Let $A$ and $B$ be square matrices with non-positive off-diagonal entries, indexed by finite sets $I_A \subseteq I_B$, and suppose $A$ is the restriction of $B$ to $I_A$. If $B$ is irreducible in the sense of Perron-Frobenius and some off-diagonal entry $B_{ij} \neq 0$ has $i$ or $j$ outside $I_A$, then $\min\sigma(B) < \min\sigma(A)$.
\end{lem}

To show monotonicity of the lowest energy in each sector with respect to the volume, we show the matrices $A = A_{L,k}$ and $B = A_{L+1, k}$ satisfy the hypothesis of the comparison lemma in the next proposition.

\begin{prop}\label{prop:comparison_lemma_hypothesis}
Let $L \geq 2$ and $1\leq k \leq \lfloor L/2 \rfloor$. (i) For $\pi'\neq \pi$,
\begin{equation}\label{eq:offdiag}
    A_{L,k}^{\pi'\pi} = \frac{(-1)^{2s}}{d}\,\# \{x : e_x\pi = \pi'\}\,.
\end{equation}
In particular the off-diagonal entry of $A_{L,k}$ is non-positive for half-odd integer $s$. (ii) Define $\iota: B_{L,k}\to B_{L+1,k}$ by letting $\iota(\pi)$ be $\pi$ with an additional dot at the site $L+1$. $\iota$ is injective and
\begin{equation}\label{eq:embedding}
    A_{L+1,k}^{\iota(\pi')\iota(\pi)} = A_{L,k}^{\pi'\pi} \qquad \text{ for all } \pi,\pi' \in B_{L,k}\,.
\end{equation}
That is, $A_{L,k}$ is the restriction of $A_{L+1,k}$ to the index subset $\iota(B_{L,k})$. (iii)  $A_{L,k}$ is irreducible in the sense of Perron-Frobenius: there is no proper nonempty subset $S \subsetneq B_{L,k}$ with $A_{L,k}^{\pi'\pi} = 0$ for all $\pi \in S$ and all $\pi' \notin S$. (iv) There exist $\pi_0 \in B_{L,k}$ and $\pi_1 \in B_{L+1,k}\setminus\iota(B_{L,k})$ with $A_{L+1,k}^{\pi_1\,\iota(\pi_0)} \neq 0$. 
\end{prop}
\begin{proof}
We first prove (i). Summing Lemma~\ref{lem:local_action} over $x = 1,\dots,L-1$, each bond falls into exactly one of the five cases. Case (C0) contributes $1$ to the diagonal entry; case (C1) contributes nothing; cases (C2)-(C4) each contribute $\frac{(-1)^{2s}}{d}$ to the entry indexed by the resulting pattern.

We now prove (ii). That $\iota(\pi)$ is a valid pattern is immediate: its arcs are those of $\pi$, hence non-crossing, and every arc has both endpoints in $\{1,\dots,L\}$, so none can span the new dot at $L+1$. Injectivity follows by deleting the last site. The image of $\iota$ consists precisely of those patterns in $B_{L+1,k}$ for which the site $L+1$ is a dot. Write the abstract Hamiltonian $\sum_{x=1}^L e_x = \sum_{x=1}^{L-1} e_x  + e_L$. For $x \leq L-1$, which of the cases of Lemma~\ref{lem:local_action} applies to the bond $\{x,x+1\}$ depends only on the status of $x$ and $x+1$ and on the location of their partners if they're connected by an arc, all of which lie in $\{1,\dots,L\}$ by assumption and are unchanged by $\iota$; so the entries are the same. Moreover the resulting pattern still has a dot at $L+1$, since case (C4) relocates arcs only among $\{x,x+1,l,m\}\subseteq\{1,\dots,L\}$ and cases (C2), (C3) move a dot to a site $l \leq L$. Hence $e_x\iota(\pi) = \iota\big(e_x\pi\big)$ for $x \leq L-1$. For the remaining term, $L+1$ is a dot of $\iota(\pi)$. If $L$ is also a dot then case (C1) applies and $e_L\iota(\pi) = 0$. If instead $L$ belongs to an arc $(l,L)$ with $l<L$, then case (C3) applies and produces a pattern containing the arc $(L,L+1)$, in which $L+1$ is therefore not a dot; such a pattern lies outside $\iota(B_{L,k})$. In either case $e_L\iota(\pi)$ has no component along $\iota(B_{L,k})$, and comparing coefficients gives \eqref{eq:embedding}.

We now prove (iii). By the first part of the proposition, all contributions to the off-diagonal entries of $A_{L,k}$ carry the common factor $(-1)^{2s}/d$, so no cancellation can occur when summing up contributions from bonds forming by applications of the generators. Hence for $\pi' \neq \pi$
\begin{equation}\label{eq:entry_nonzero}
    A_{L,k}^{\pi'\pi} \neq 0
    \quad\Longleftrightarrow\quad
    e_x\pi = \pi' \ \text{ for some } x \in \{1,\dots,L-1\}\,.
\end{equation}
Suppose for the sake of contradiction $S \subsetneq B_{L,k}$ is nonempty with $A_{L,k}^{\pi'\pi} = 0$ whenever $\pi \in S$ and $\pi' \notin S$, and let $M = \mathrm{span}\{\pi : \pi \in S\} \subseteq W_{L,k}$. For $\pi \in S$ and any $x$, we have $e_x\pi$ either equals $\pi, 0$ or $\pi'$ for some other configuration $\pi'$. \eqref{eq:entry_nonzero} and our assumption on $S$ implies we must have $\pi' \in S$ in the third case. Hence $e_x\pi \in M$ in every case, for all $x$. Thus, $M$ is a nonzero proper $TL_L(d)$-submodule of $W_{L,k}$, contradicting the simplicity of the standard module.

We now prove (iv). Since $k \geq 1$ and $2k \leq L$, the pattern
\[
    \pi_0 = \{(L-2k+1,\,L-2k+2),\,(L-2k+3,\,L-2k+4)\,\cdots(L-1,\,L)\}
\]
lies in $B_{L,k}$: its arcs are disjoint and non-crossing, and all $L-2k$ dots precede all arcs, so no arc spans a dot. Its site $L$ belongs to the arc $(L-1,L)$, while $L+1$ is a dot of $\iota(\pi_0)$. Case (C3) of Lemma~\ref{lem:local_action} therefore applies to the bond $\{L,L+1\}$ and produces the pattern
\[
    \pi_1 = \big(\iota(\pi_0)\setminus\{(L-1,L)\}\big)\cup\{(L,L+1)\}
\]
together with a dot at $L-1$, with coefficient $(-1)^{2s}/d \neq 0$. Since $L+1$ is not a dot of $\pi_1$, we have $\pi_1 \notin \iota(B_{L,k})$, and $A_{L+1,k}^{\pi_1\,\iota(\pi_0)}\neq 0$ is an off-diagonal entry of $A_{L+1,k}$ one of whose indices lies outside $\iota(B_{L,k})$.
\end{proof}

\begin{cor}[Monotonicity in the volume]\label{cor:monotonicity}
Let $s$ be a half-odd integer. Then for all $k \geq 1$ and $L \geq 2$ with $W_{L,k}\neq 0$, we have $\cE(L+1,k) < \cE(L,k)$.
\end{cor}
\begin{proof}
Combine Lemma~\ref{lem:ns73} with Proposition~\ref{prop:comparison_lemma_hypothesis}.
\end{proof}

Note that the above statement fails at $k=0$, where both sides vanish; this case is handled separately below and causes no difficulty.

\subsection{From Monotonicity in Volume to Symmetry Sectors}
\label{sec:upgrading}

The passage from Corollary~\ref{cor:monotonicity} to Theorem~\ref{thm:foel} uses only positivity of the interaction, and a branching rule, that is insensitive to the sign discussion above. The role played by the Clebsch-Gordan decomposition in \cite{nachtergaele2004ferromagnetic} is played here by the restriction of the standard modules to a subchain \cite{aufgebauer2010quantum, goodman1989coxeter, westbury1995representation},
\begin{equation}\label{eq:branching}
    W_{L+1,k}\big|_{TL_L(d)} \cong
    \begin{cases}
        W_{L,0} & k = 0,\\[2pt]
        W_{L,k-1} & k = \tfrac{L+1}{2},\\[2pt]
        W_{L,k}\oplus W_{L,k-1} & \text{otherwise.}
    \end{cases}
\end{equation}
This is again a two-term decomposition when $k$ is not an edge case: the last site of a pattern in $B_{L+1,k}$ is either a dot, in which case deleting it gives an element of $B_{L,k}$, or the right endpoint of an arc, in which case deleting that arc frees its partner into a dot and gives an element of $B_{L,k-1}$. The arguments in the following proposition and its corollary follow those in \cite{nachtergaele2004ferromagnetic}.

\begin{prop}\label{prop:two_term}
For all $L \geq 2$ and $k \geq 0$,
\begin{equation}\label{eq:two_term}
    \cE(L+1,k) \geq \min\{\cE(L,k),\, \cE(L,k-1)\}\,.
\end{equation}
\end{prop}
\begin{proof}
Identify $\cH_{L+1} = \cH_L \otimes V_s$ and let $T L_L(d)$ act on the first factor. For each $j$ write $\cK_j \subseteq \cH_{L+1}$ for the $j$-isotypic component of this action, and write $\cH_{L+1}^{(k)}$ for the $k$-isotypic component of the action of $TL_{L+1}(d)$ on $\cH_{L+1}$. By \eqref{eq:branching} the restriction of $W_{L+1,k}$ to $TL_L(d)$ involves only the types $k$ and $k-1$, whence
\[
    \cH_{L+1}^{(k)} \;\subseteq\; \cK_k \oplus \cK_{k-1}\,.
\]
The representation $\rho_d$ is a $*$-representation, by self-adjointness of the Temperley-Lieb representatives $d\,P^{(0)}_{x,x+1}$. Consequently the orthogonal complement of an invariant subspace is again invariant, so the orthogonal projection onto an isotypic component commutes with the action; its restriction to any other isotypic component is then a module map between isotypic components of non-isomorphic simples, hence zero by Schur's lemma. In other words $\cK_k$ and $\cK_{k-1}$ are orthogonal. Moreover $H_L \otimes \ds$ preserves each $\cK_j$ and acts there with $\min\sigma\big(H_L\otimes\ds\big|_{\cK_j}\big) = \cE(L,j)$.

Let $\psi$ be a normalized vector in $\cH_{L+1}^{(k)}$ and decompose $\psi = \psi_k + \psi_{k-1}$ accordingly. Since the two summands are orthogonal and $H_L \otimes \ds$ preserves each $\cK_j$, the cross terms in $\langle\psi, H_L\psi\rangle$ vanish and $\|\psi_k\|^2 + \|\psi_{k-1}\|^2 = 1$. As $H_{L+1} = H_L + P^{(0)}_{L,L+1} \geq H_L$ by positivity of the interaction,
\[
    \langle \psi, H_{L+1}\psi\rangle \ \geq\ \langle \psi_k, H_L\psi_k\rangle
    + \langle \psi_{k-1}, H_L\psi_{k-1}\rangle
    \ \geq\ \cE(L,k)\|\psi_k\|^2 + \cE(L,k-1)\|\psi_{k-1}\|^2\,,
\]
and the right hand side being at least $\min\{\cE(L,k),\,\cE(L,k-1)\}$.

In the two degenerate cases of \eqref{eq:branching} one of the summands is zero, so the corresponding component of $\psi$ vanishes and only one term survives: at $k=0$ one has $\psi_{-1}=0$ and the bound reads $\cE(L+1,0)\geq\cE(L,0)$, while at $k=\tfrac{L+1}{2}$ one has $\psi_k=0$ and it reads $\cE(L+1,k)\geq\cE(L,k-1)$. With the convention $\cE(L,k)=+\infty$ when $W_{L,k}=0$, both are again instances of \eqref{eq:two_term}.
\end{proof}

\begin{cor}\label{cor:propagation}
Fix $k$ and suppose $\cE(L,k) < \cE(L,r)$ for all $r>k$. If moreover $\cE(L+1,k)<\cE(L,k)$, then $\cE(L+1,k) < \cE(L+1,r)$ for all $r > k$.
\end{cor}
\begin{proof}
Let $r>k$ so $r-1 \geq k$. Then $\cE(L,r) > \cE(L,k)$ by hypothesis and $\cE(L,r-1)\geq \cE(L,k)$, with equality possible only when $r-1 = k$. Proposition~\ref{prop:two_term} then gives
\[
    \cE(L+1,r) \ \geq\ \min\{\cE(L,r),\cE(L,r-1)\} \ \geq\ \cE(L,k) \ >\ \cE(L+1,k)\,. \qedhere
\]
\end{proof}

We can now prove FOEL for half-odd integer spin-$s$.
\begin{proof}[Proof of Theorem~\ref{thm:foel} for half-odd integer $s$]
We induct on $L$. For $L=2$ we have $\cH_2 = \big(W_{2,0}\otimes U_0\big)\oplus\big(W_{2,1}\otimes U_1\big)$, the two standard modules being spanned by the pattern with two dots and the pattern with a single arc respectively. The Hamiltonian $H_2 = P^{(0)}_{1,2}$ vanishes on the first and acts as the identity on the second, so $\cE(2,0)=0 < 1 = \cE(2,1)$.

Assume \eqref{eq:foel} holds at $L$ and fix $k$ with $W_{L+1,k}\neq 0$. If $W_{L,k}\neq 0$ as well, then the inductive hypothesis gives $\cE(L,k)<\cE(L,r)$ for all $r>k$, Corollary~\ref{cor:monotonicity} gives $\cE(L+1,k)<\cE(L,k)$ when $k \geq 1$, and Corollary~\ref{cor:propagation} yields $\cE(L+1,k)<\cE(L+1,r)$ for all $r>k$. For $k=0$ one argues directly: $\cE(L+1,0)=0$ while $\cE(L+1,r)>0$ for $r \geq 1$ by frustration freeness, since the kernel of $H_{L+1}$ is exactly the sector $k=0$.

It remains to treat the case $W_{L,k}=0$, which occurs only for $k = \tfrac{L+1}{2}$, the sector newly available at length $L+1$. Since this is the largest admissible index at length $L+1$, there is no $r>k$ to consider, and with the convention $\cE(L,k)=+\infty$ Proposition~\ref{prop:two_term} gives $\cE(L+1,k)\geq \cE(L,k-1)$. For $1 \leq r < k$ we then have, using the inductive hypothesis and Corollary~\ref{cor:monotonicity},
\[
    \cE(L+1,r) \ <\ \cE(L,r) \ \leq\ \cE(L,k-1) \ \leq\ \cE(L+1,k)\,,
\]
while $\cE(L+1,0) = 0 < \cE(L+1,k)$ by frustration freeness.
\end{proof}

\section{The Spectral Gap}
\label{sec:gap}

As a consequence of the ordering of energy levels, we can derive a uniform lower bound on the finite volume gaps. Theorem~\ref{thm:foel} reduces the computation of the spectral gap of $H_L$ to the diagonalization of the $k=1$ sector, and that sector turns out to be a nearest-neighbor hopping problem on a path. The proof follows closely that of \cite{nachtergaele2004ferromagnetic}.

\begin{proof}[Proof of Corollary~\ref{cor:fv_gap}]
By frustration freeness the ground state energy is $0$ and the ground state space is the sector $k=0$, so Theorem~\ref{thm:foel} gives $\gamma_L = \cE(L,1)-\cE(L,0) = \cE(L,1)$.

A pattern with a single arc has $L-2$ dots, and since no arc may span a dot the arc must join two adjacent sites. Hence
\[
    B_{L,1} = \{\pi_x : x = 1,\dots,L-1\}\,, \qquad \pi_x = \text{ the pattern whose only arc is } (x,x+1)\,,
\]
so that $\dim W_{L,1}=L-1$. Fix $\ket{V}\in U_1$ and abbreviate $\Psi_x = \Psi_{\pi_x,V}$. By Lemma~\ref{lem:local_action}, the bond $\{x,x+1\}$ falls into case (C0), the bonds $\{x-1,x\}$ and $\{x+1,x+2\}$ fall into cases (C3) and (C2) respectively, and every other bond joins two dots and falls into case (C1). Summing over all bonds gives
\begin{equation}\label{eq:one_arc_action}
    H_L \Psi_x = \Psi_x + \frac{(-1)^{2s}}{d}\big(\Psi_{x-1}+\Psi_{x+1}\big)\,,
\end{equation}
for each $1 \leq x \leq L-1$, where we adopt the convention $\Psi_0 = \Psi_L = 0$. Therefore
\begin{equation}\label{eq:AL1}
    A_{L,1} = \ds + \frac{(-1)^{2s}}{d}\,\Delta_{L-1}\,,
\end{equation}
where $\Delta_{L-1}$ is the adjacency matrix of the path on $L-1$ vertices. Its spectrum is $\{2\cos(\pi m/L) : m=1,\dots,L-1\}$, so
\begin{equation}\label{eq:spec_AL1}
    \sigma(A_{L,1}) = \Big\{1+(-1)^{2s}\cdot\tfrac{2}{d}\cos\big(\tfrac{\pi m}{L}\big) : m = 1,\dots,L-1\Big\}\,,
\end{equation}
which is minimized at $m=1$ when $s$ is a half-odd integer, and minimized at $m=L-1$ when $s$ is an integer. In both cases, $\cE(L,1) = 1-\tfrac2d\cos(\pi/L)$ and $\inf_L\gamma_L = 1 - \tfrac2d > 0$.
\end{proof}

\begin{remark}\label{rem:parity}
The sign in front of $\cos k$ in Equation~\eqref{eq:spec_AL1} is the same one appearing in Proposition~\ref{thm:spinwave_energies} and Equation~\eqref{eq:band_energy}, all for the same reason. For integer $s$ the singlet subspace $V_0 \subset V_s\otimes V_s \cong V_{2s} \oplus V_{2s-1} \oplus \cdots \oplus V_0$ is symmetric under exchange, like $V_{2s}$; for half-odd integer $s$ it is antisymmetric, like $V_{2s-1}$. This is the content of the reversal rule \eqref{eq:reversal}, and it explains both the $1+\cos k$ accompanying $K_{2s}$ and the $1-\cos k$ accompanying $K_{2s-1}$ in Proposition~\ref{thm:spinwave_energies}. Since the band $\{1+(-1)^{2s}2\cos k/d\}_{k}$ is the same set for both parities, the value of the gap is unaffected.
\end{remark}

The rest of this section is devoted to the proof of Theorem~\ref{thm:gns_gap}, deriving the spectral gap of $\omega$'s GNS Hamiltonian with respect to $\tau_t$. In contrast with the above, the argument we present does not use the ordering of energy levels. Instead, the upper and lower bounds on the gap are obtained independently, in Sections~\ref{sec:upper} and~\ref{sec:lower}, and are found to coincide. Both bounds pass through the periodic Hamiltonians $H_L^{\text{PBC}}$. The results in Section~\ref{sec:spin_system_facts} of the appendix allow information about the Hamiltonians with periodic boundaries to be transferred to the GNS Hamiltonian $H$, which was defined using the interaction $\Phi$ implementing open boundary conditions. 

\subsection{The Upper Bound}
\label{sec:upper}

We exhibit an explicit band of excited states of the periodic chains and show that it survives the infinite volume limit as essential spectrum of $H$. In the spin-$1$ case these states, together with the higher deviate states, were used by \cite{parkinson1988} and \cite{koberle1994deformed} to solve the spin-$1$ biquadratic model and its $q$-deformation by Bethe ansatz. The vectors below are their higher spin generalizations, which also appeared in \cite{koberle1995AD} and \cite{limasantos1998exact} to solve more general Temperley--Lieb chains with Bethe ansatz as well. 

Fix $L \geq 3$ and let $\ket{-s}_x \in \cH_L$ denote the vector that is $\ket{s}$ at every site except $x$, where it is $\ket{-s}$; for $\ga \neq \pm s$ let $\ket{\ga,-\ga}_{x,x+1}$ denote the vector that is $\ket{s}$ at every site except $x$ and $x+1$, which are $\ket{\ga}$ and $\ket{-\ga}$ respectively. At $\ga = \pm s$ the two-site excitation becomes $\ket{-s}_{x+1}$ or $\ket{-s}_x$, hence the restriction $\gamma \neq \pm s$. Let $m : \{-s,\dots,s-1\}\to\{1,\dots,2s\}$ be the bijection $i \mapsto i+s+1$. For $k = \tfrac{2\pi}{L}j$ with $j=0,1,\dots,L-1$, define vectors
\begin{equation}\label{eq:basis_higher_spin}
    v_1^k = v_{m_{-s}}^k = \sum_{x=1}^{L}e^{ikx}\ket{-s}_x \qqtext{ and }
    v_{m_{\ga}}^k = \sum_{x=1}^{L}e^{ikx}\ket{\ga,-\ga}_{x,x+1}\,,
\end{equation}
where $\ga \in \{-s+1,\dots,s-1\}$. For $s=1$ these are simply $v_1^k = \sum_x e^{ikx}\ket{-1}_x$ and $v_2^k = \sum_x e^{ikx}\ket{00}_{x,x+1}$. The quantization of the momenta $k$ is forced by translation invariance of the periodic chain. The vectors \eqref{eq:basis_higher_spin} are mutually orthogonal for each fixed $k$ and span a $2s$-dimensional subspace $Y_k \subset \cH_L$ that is invariant under $H_L^{\text{PBC}}$; we diagonalize $H_L^{\text{PBC}}$ in each $k$-sector separately, by diagonalizing the matrix representation $A_k$ of $H_L^{\text{PBC}}$ in the subspace $Y_k$ using the described basis.

\begin{prop}\label{prop:band}
Let $L \geq 3$, and $k = \tfrac{2\pi}{L}j$ with $j \in \{0,\dots,L-1\}$. Then $Y_k$ is invariant under $H^{\text{PBC}}_L$, the restriction $H^{\text{PBC}}_L|_{Y_k}$ has rank one, and its unique nonzero eigenvalue is
\begin{equation}\label{eq:band_energy}
    \lambda(k) = 1 + (-1)^{2s}\,\frac{2\cos k}{d}\,.
\end{equation}
In particular $\lambda(k) \in \sigma(H_L^{\text{PBC}})$ for every such $k$.
\end{prop}
\begin{proof}
Recall from \eqref{eq:projector_P0} that $d\,P^{(0)} = \ket{\psi}\bra{\psi}$ with $\braket{\psi}{\al,-\al} = (-1)^{s-\al}$, and that $P^{(0)}$ annihilates $\ket{s}\otimes\ket{s}$. Hence only the bonds $\{x-1,x\}$ and $\{x,x+1\}$ act nontrivially on $\ket{-s}_x$, with $\braket{\psi}{s,-s}=1$ and $\braket{\psi}{-s,s} = (-1)^{2s}$, and only the bond $\{x,x+1\}$ acts nontrivially on $\ket{\ga,-\ga}_{x,x+1}$. Therefore
\begin{align}
    d\,H_L^{\text{PBC}}\,v_1^k
    &= \sum_{x=1}^{L}e^{ikx}\Big[dP^{(0)}_{x-1,x}+dP^{(0)}_{x,x+1}\Big]\ket{-s}_x \nonumber\\
    &= \sum_{x=1}^{L}e^{ikx}\left[\sum_{\al=-s}^{s}(-1)^{s-\al}\ket{\al,-\al}_{x-1,x}
       + \sum_{\al=-s}^{s}(-1)^{3s-\al}\ket{\al,-\al}_{x,x+1}\right] \nonumber\\
    &= \sum_{x=1}^{L}e^{ikx}\Big[2\ket{-s}_x + (-1)^{2s}\ket{-s}_{x-1} + (-1)^{2s}\ket{-s}_{x+1}\Big]
       \nonumber\\
    &\qquad + \sum_{\al=-s+1}^{s-1}(-1)^{s-\al}\sum_{x=1}^{L}e^{ikx}
       \Big[\ket{\al,-\al}_{x-1,x} + (-1)^{2s}\ket{\al,-\al}_{x,x+1}\Big] \nonumber\\
    &= \Big[2+(-1)^{2s}\big(e^{ik}+e^{-ik}\big)\Big]v_1^k
       + \sum_{\al=-s+1}^{s-1}(-1)^{s-\al}\Big[e^{ik}+(-1)^{2s}\Big]v^k_{m_\al} \nonumber\\
    &= 2\big(1+(-1)^{2s}\cos k\big)\,v_1^k
       + \big((-1)^{2s}+e^{ik}\big)\sum_{\al=-s+1}^{s-1}(-1)^{s-\al}\,v^k_{m_\al}\,,
    \label{eq:action_v1}
\end{align}
where in the third equality the terms $\al = \pm s$ of the two sums were separated off, these being the one-site excitations $\ket{s,-s}_{x-1,x} = \ket{-s}_x$, $\ket{-s,s}_{x-1,x} = \ket{-s}_{x-1}$, $\ket{s,-s}_{x,x+1} = \ket{-s}_{x+1}$ and $\ket{-s,s}_{x,x+1} = \ket{-s}_x$, and in the fourth equality we reindexed using
\begin{equation}\label{eq:shifts}
    \sum_{x=1}^{L}e^{ikx}\ket{-s}_{x\mp1} = e^{\pm ik}\,v_1^k \qqtext{ and }
    \sum_{x=1}^{L}e^{ikx}\ket{\al,-\al}_{x-1,x} = e^{ik}\,v^k_{m_\al}\,.
\end{equation}
In the same way,
\begin{align}
    d\,H_L^{\text{PBC}}\,v^k_{m_\ga}
    &= \sum_{x=1}^{L}e^{ikx}\,dP^{(0)}_{x,x+1}\ket{\ga,-\ga}_{x,x+1}
     = (-1)^{s-\ga}\sum_{x=1}^{L}e^{ikx}\sum_{\al=-s}^{s}(-1)^{s-\al}\ket{\al,-\al}_{x,x+1} \nonumber\\
    &= (-1)^{s-\ga}\sum_{x=1}^{L}e^{ikx}\left[(-1)^{2s}\ket{-s}_x + \ket{-s}_{x+1}
       + \sum_{\al=-s+1}^{s-1}(-1)^{s-\al}\ket{\al,-\al}_{x,x+1}\right] \nonumber\\
    &= (-1)^{s-\ga}\big((-1)^{2s}+e^{-ik}\big)\,v_1^k
       + (-1)^{s-\ga}\sum_{\al=-s+1}^{s-1}(-1)^{s-\al}\,v^k_{m_\al}\,.
    \label{eq:action_vga}
\end{align}
In particular $Y_k$ is invariant under $H_L^{\text{PBC}}$.

Now set
\[
    a = (-1)^{2s} + e^{ik}\,, \qquad \bar a = (-1)^{2s} + e^{-ik}\,, \qquad
    u_\al = (-1)^{s-\al} \quad (\al = -s+1,\dots,s-1)\,,
\]
so that $u = (u_{-s+1},...,u_{s-1})^T \in \rr^{2s-1}$ and
\begin{equation}\label{eq:a_abar}
    a\bar a = 1 + (-1)^{2s}\big(e^{ik}+e^{-ik}\big) + 1 = 2\big(1+(-1)^{2s}\cos k\big)\,.
\end{equation}
Let $A_k$ denote the matrix of $H_L^{\text{PBC}}|_{Y_k}$ in the ordered basis $\big(v_1^k, v^k_{m_{-s+1}},\dots,v^k_{m_{s-1}}\big)$ of \eqref{eq:basis_higher_spin}. Reading the coefficients off \eqref{eq:action_v1} and \eqref{eq:action_vga}, and rewriting the $(1,1)$ entry by means of \eqref{eq:a_abar}, we have $dA_k$ is the outer product
\begin{equation}\label{eq:rank_one}
    dA_k
    = \begin{pmatrix} a\bar a & \bar a\,u^{T}\\[2pt] a\,u & u\,u^{T}\end{pmatrix}
    = p\,q^{T}\,,
    \qquad
    p = \begin{pmatrix} \bar a\\ u\end{pmatrix}\,,
    \qquad
    q = \begin{pmatrix} a\\ u\end{pmatrix}\,.
\end{equation}
In particular $d\,A_k$ has rank one and its kernel has dimension $2s-1$. Since $(pq^{T})p = (q^{T}p)\,p$, the unique nonzero eigenvalue is $q^{T}p$, and as $u_\al^2 = 1$ for each of the $2s-1$ admissible values of $\al$,
\[
    q^{T}p = a\bar a + \sum_{\al=-s+1}^{s-1}u_\al^{2}
    = 2\big(1+(-1)^{2s}\cos k\big) + (2s-1)
    = (2s+1) + 2(-1)^{2s}\cos k
    = d + 2(-1)^{2s}\cos k\,.
\]
Dividing by $d$ yields \eqref{eq:band_energy}.
\end{proof}

The above band of energies can be constructed for $H_{\Lambda_n}^{\Phi_n}$ on the chain $\Lambda_n = [-L_n, L_n]$ as well, with the same eigenvalues. The only difference is the momenta are now discretized as $\{\frac{2\pi}{N_n}j: j=0,...,N_n-1\}$ (recall $N_n = 2L_n+1$), due to the number of sites, not the position on the chain. We consider these chains in the following propositions because we want to discuss the thermodynamic limit on $\zz$ and not $\nn$.

\begin{prop}\label{prop:essential_spectrum}
The essential spectrum of the GNS Hamiltonian $H$ contains the interval $I = \big[1-\tfrac2d,\, 1+\tfrac2d\big]$.
\end{prop}
\begin{proof}
We construct a Weyl sequence for each $E \in I$. That is, we produce unit vectors $\psi_n \in \cH_\omega$ with $\psi_n \rightharpoonup 0$ and $\|(H-E)\psi_n\|\to0$. Given $E \in I$ choose $k \in [0,2\pi]$ with $E = \lambda(k)$, where $\lambda(k)$ is given in \eqref{eq:band_energy}. Choose a sequence $k_n \to k$ with each $k_n = \tfrac{2\pi}{N_n}j_n$ for some $j_n \in \{0,\dots,N_n-1\}$, so that $\lambda_n \coloneq \lambda(k_n) = 1 + (-1)^{2s} \,\tfrac{2\cos k_n}{d}$ is an eigenvalue of $H_{\Lambda_n}^{\Phi_n}$ as in Proposition~\ref{prop:band}. Let $\widehat\Lambda_n = [-L_n-1,L_n+1]$ be $\Lambda_n$ together with its two neighboring sites. Throughout, $h_{x,y} = P^{(0)}_{x,y}$ and $\ket{s\cdots s}_\Lambda$ denotes the fully polarized vector of $\cH_\Lambda$.

For $A \in \cA_{\text{loc}}$ with $\mathrm{supp}\,A \subseteq \Lambda_n$ one has $H\pi(A)\Omega = \pi(\delta(A))\Omega$, where
\[
    \delta(A) = \sum_{X :\, X \cap \mathrm{supp}(A) \neq \emptyset}[\Phi(X),A] = [H_{\widehat\Lambda_n},\, A]\,,
    \qquad
    H_{\widehat\Lambda_n} = \sum_{x=-L_n-1}^{L_n} h_{x,x+1}\,.
\]
Since $H_{\widehat\Lambda_n}\ket{s\cdots s}_{\widehat\Lambda_n} = 0$ we have $\pi(AH_{\widehat\Lambda_n})\Omega = 0$. Then for every $E \in \rr$,
\begin{equation}\label{eq:reduction}
    \big\|(H-E)\pi(A)\Omega\big\|
    = \big\|\pi\big((H_{\widehat\Lambda_n}-E)A\big)\Omega\big\|
    = \big\|(H_{\widehat\Lambda_n}-E)\,A\ket{s\cdots s}_{\widehat\Lambda_n}\big\|\,,
\end{equation}
the last equality because $(H_{\widehat\Lambda_n}-E)A \in \cA_{\widehat\Lambda_n}$ and
$\omega$ restricts on $\cA_{\widehat\Lambda_n}$ to the vector state $\ket{s\cdots s}_{\widehat\Lambda_n}$. Likewise $\|\pi(A)\Omega\|
= \|A\ket{s\cdots s}_{\widehat\Lambda_n}\|$. By Proposition~\ref{prop:band} the eigenvector of $H^{\Phi_n}_{\Lambda_n}$ for the eigenvalue $\lambda_n$ lies in $Y_{k_n}$; since the $v^{k}_i$ of \eqref{eq:basis_higher_spin} are mutually orthogonal with $\|v^{k}_i\|^2 = N_n$ (recall these vectors are now being taken on the chain $\Lambda_n$), we may write the eigenvector in normalized form as
\begin{equation}\label{eq:chi_n}
    \chi_n \;\coloneq\; \frac{1}{\sqrt{N_n}}\Big(c_1 v_1^{k_n} + \sum_{\ga} c_{m_\ga}v^{k_n}_{m_\ga}\Big)
    \;\in\; \cH_{\Lambda_n}\,,
    \qquad
    |c_1|^2 + \sum_\ga |c_{m_\ga}|^2 = 1\,.
\end{equation}
For $x = -L_n,\dots,L_n$, set $\cO^{-s}_x = \big(\ket{-s}\bra{s}\big)_x$. For $\ga \in \{-s+1,\dots,s-1\}$ set $\cO^{\ga}_x = \big(\ket{\ga}\bra{s}\big)_x \otimes \big(\ket{-\ga}\bra{s}\big)_{x+1}$ for $-L_n \le x \le L_n-1$ and $\cO^{\ga}_{L_n} = \big(\ket{\ga}\bra{s}\big)_{L_n} \otimes \big(\ket{-\ga}\bra{s}\big)_{-L_n}$ for the wrap-around bond $\{L_n,-L_n\}$ of the ring. Writing $\cO_x = c_1\cO^{-s}_x + \sum_\ga c_{m_\ga}\cO^{\ga}_x$ and
\begin{equation}\label{eq:An}
    A_n = \frac{1}{\sqrt{N_n}}\sum_{x=-L_n}^{L_n} e^{ik_n x}\,\cO_x \ \in\ \cA_{\Lambda_n}\,,
\end{equation}
one has $A_n\ket{s\cdots s} = \chi_n$ exactly. In particular $\|\pi(A_n)\Omega\| = 1$, and $\|\cO_x\| \le |c_1| + \sum_\ga|c_{m_\ga}| \le \sqrt{2s}$ by Cauchy-Schwarz, for every $x$.

Comparing the ring Hamiltonian with $H_{\widehat\Lambda_n}$,
\[
    H_{\widehat\Lambda_n} = H^{\Phi_n}_{\Lambda_n} - h_{L_n,-L_n} + h_{-L_n-1,-L_n} + h_{L_n,L_n+1}\,,
\]
so by $H^{\Phi_n}_{\Lambda_n}\chi_n = \lambda_n \chi_n$ and \eqref{eq:reduction},
\begin{equation}\label{eq:three_terms}
    \big\|(H-\lambda_n)\pi(A_n)\Omega\big\|
    \;\le\; \|h_{L_n,-L_n}\chi_n\| + \|h_{-L_n-1,-L_n}\chi_n\| + \|h_{L_n,L_n+1}\chi_n\|\,,
\end{equation}
where on the right $\chi_n$ is understood as $\ket{s}_{-L_n-1} \otimes \chi_n \otimes \ket{s}_{L_n+1} \in \cH_{\widehat\Lambda_n}$. We bound the three terms.

Let $T$ be the cyclic shift on $\cH_{\Lambda_n}$, so that $T\ket{-s}_x = \ket{-s}_{x+1}$ and $T\ket{\ga,-\ga}_{x,x+1} = \ket{\ga,-\ga}_{x+1,x+2}$ modulo $N_n$. Then $Tv^{k_n}_i = e^{-ik_n}v^{k_n}_i$ for every basis vector of \eqref{eq:basis_higher_spin}, hence $T\chi_n = e^{-ik_n}\chi_n$, while $T h_{x,x+1}T^* = h_{x+1,x+2}$ modulo $N_n$. Therefore $\langle \chi_n, h_{x,x+1}\chi_n\rangle$ is independent of $x$, and summing over the $N_n$ bonds of the ring gives $\langle \chi_n, H^{\Phi_n}_{\Lambda_n}\chi_n\rangle = \lambda_n$. As $h_{L_n,-L_n}$ is a projection,
\begin{equation}\label{eq:wrap}
    \|h_{L_n,-L_n}\chi_n\|^2 = \langle\chi_n, h_{L_n,-L_n}\chi_n\rangle = \frac{\lambda_n}{N_n}\,.
\end{equation}
For the bond $\{-L_n-1,-L_n\}$, recall $dh_{-L_n-1,-L_n} = \ket{\psi}\bra{\psi}_{-L_n-1,-L_n}$ with $\braket{\psi}{\al,-\al} = (-1)^{s-\al}$, so that $\bra{\psi}_{-L_n-1,-L_n}$ annihilates $\ket{s}_{-L_n-1}\otimes\ket{\beta}_{-L_n}$ unless $\beta = -s$, in which case it gives $1$. Among the summands of \eqref{eq:chi_n} only the $x=-L_n$ term of $v^{k_n}_1$ carries $\ket{-s}$ at site $-L_n$. Hence $\bra{\psi}_{-L_n-1,-L_n}\big(\ket{s}_{-L_n-1}\otimes\chi_n\big) = N_n^{-1/2}c_1e^{-ik_nL_n}\ket{s\cdots s}_{[-L_n+1,L_n]}$ and, using $\|\ket{\psi}\|^2 = d$,
\begin{equation}\label{eq:left_edge}
    \|h_{-L_n-1,-L_n}\chi_n\|^2 = \frac{1}{d^2}\,\big\|\ket{\psi}\big\|^2\,\frac{|c_1|^2}{N_n}
    = \frac{|c_1|^2}{dN_n} \;\le\; \frac{1}{dN_n}\,.
\end{equation}
The bond $\{L_n,L_n+1\}$ is similar: $\bra{\psi}_{L_n,L_n+1}$ annihilates $\ket{\beta}_{L_n}\otimes\ket{s}_{L_n+1}$ unless $\beta = -s$, in which case it gives $(-1)^{2s}$, and again only the $x=L_n$ term of $v^{k_n}_1$ carries $\ket{-s}$ at site $L_n$. Thus
\begin{equation}\label{eq:right_edge}
    \|h_{L_n,L_n+1}\chi_n\|^2 = \frac{|c_1|^2}{dN_n} \le \frac{1}{dN_n}\,.
\end{equation}
Substituting \eqref{eq:wrap}, \eqref{eq:left_edge} and \eqref{eq:right_edge} into \eqref{eq:three_terms} and using
$\lambda_n \le 1 + \tfrac2d \le 2$,
\[
    \big\|(H-\lambda_n)\pi(A_n)\Omega\big\|
    \;\le\; \frac{1}{\sqrt{N_n}}\Big(\sqrt{2} + \frac{2}{\sqrt{d}}\Big)\,.
\]

Define the Weyl sequence by $\psi_n \coloneq \pi(A_n)\Omega$. First, note each element is a unit vector. Second,
\[
    \|(H-E)\psi_n\| \;\le\; \big\|(H-\lambda_n)\psi_n\big\| + |\lambda_n - E|
    \;\longrightarrow\; 0\,,
\]
since $\lambda_n \to E$. It finally remains to check $\psi_n \rightharpoonup 0$. Let $\psi' = \pi(B)\Omega$ with $B \in \cA_{\text{loc}}$ and $n$ large enough so $\mathrm{supp}(B) \subset \Lambda_n$. Consider
\[
    \bra{\psi'}\ket{\psi_n} = \omega(B^*A_n)
    = \frac{1}{\sqrt{N_n}}\sum_{x=-L_n}^{L_n}e^{ik_nx}\,\omega(B^*\cO_x)\,.
\]
If $\mathrm{supp}\,\cO_x \cap \mathrm{supp}\,B = \emptyset$ the summand factorizes as
$\omega(B^*)\,\omega(\cO_x)$ and vanishes, because $\omega(\cO_x) = 0$. Since each $\cO_x$ is supported on two sites, at most $2|\mathrm{supp}\,B| + 2$ summands survive, whence
\[
    \big|\bra{\psi'}\ket{\psi_n}\big|
    \;\le\; \frac{\big(2|\mathrm{supp}\,B|+2\big)\,\|B\|\,\sqrt{2s}}{\sqrt{N_n}}
    \;\longrightarrow\; 0\,.
\]
As $\pi(\cA_{\text{loc}})\Omega$ is dense in $\cH_\omega$, $\psi_n \rightharpoonup 0$.
\end{proof}

\subsection{The Lower Bound}
\label{sec:lower}

The matching lower bound comes from Knabe's finite-size criterion \cite{knabe1988energy}, which converts the spectral gap of a small open chain into a bound on the gap of arbitrarily long periodic chains.

\begin{thm}[{\cite{knabe1988energy}}]\label{thm:knabe}
Let $H_K = \sum_{x=1}^{L}h_{x,x+1}$ be a translation invariant frustration free Hamiltonian on a periodic chain of length $L$, whose interaction term is an orthogonal projection, and let $\epsilon_m$ denote the spectral gap of the same model on an open chain of $m+1$ sites. Then for $1 < m < \lfloor L/2\rfloor$,
\begin{equation}\label{eq:knabe}
    H_K^2 \geq \eta_m H_K\,, \qquad \eta_m = \frac{m}{m-1}\Big(\epsilon_m - \frac1m\Big)\,.
\end{equation}
\end{thm}

Since $\eta_m$ does not depend on $L$, any $m$ with $\epsilon_m > 1/m$ produces an $L$-independent lower bound on the finite volume gap of the singlet model with periodic boundaries, when the theorem is applied to $H_K = H_L^{\text{PBC}}$. It turns out, the optimal such bound is obtained by taking $m=2$ for every spin $s$, as seen in the next proposition. Larger values of $m$ are admissible but give weaker bounds, as Table~\ref{tab:knabe_bulk_gap} shows; the criterion is sharp exactly at $m=2$.

\begin{table}[h!]
    \centering
    \renewcommand{\arraystretch}{1.3}
    \begin{tabular}{c | c c c}
        \hline\hline
        \textbf{Spin} $s$ & \boldmath$m=2$ & \boldmath$m=3$ & \boldmath$m=4$ \\
        \hline
        $1$   & $0.333333$ & $0.292893$ & $0.280874$ \\
        $3/2$ & $0.500000$ & $0.469670$ & $0.460655$ \\
        $2$   & $0.600000$ & $0.575736$ & $0.568524$ \\
        $5/2$ & $0.666667$ & $0.646447$ & --         \\
        \hline\hline
    \end{tabular}
    \caption{Lower bounds $\eta_m$ on the bulk gap obtained from Knabe's criterion evaluated at $m$ bonds. The column $m=2$ reproduces $1-\tfrac2d$ exactly.}
    \label{tab:knabe_bulk_gap}
\end{table}

\begin{prop}\label{prop:three_site}
The spectrum of the singlet model on three sites with open boundaries is
\begin{equation}\label{eq:three_site_spectrum}
    \sigma(H_3) = \Big\{0,\ 1-\tfrac1d,\ 1+\tfrac1d\Big\}\,.
\end{equation}
\end{prop}
\begin{proof}
By the equivalent formulation of the singlet model in terms of the $O(d)$ singlet \eqref{eq:Od_singlet}, we may compute with the unitarily transformed Hamiltonian $dH_3 = E_1+E_2$ where $E_x = \ket{\phi}\bra{\phi}$ acts on the bond $\{x,x+1\}$. Here $\ket{\phi} = \sum_{j=1}^d\ket{j,j}$, for a basis of the onsite Hilbert space $\{\ket{j} : j=1,...d\}$. For each $i=1,\dots,d$ define
\[
    v_1^i = \sum_j \ket{j,j,i}\,, \qquad v_2^i = \sum_j\ket{i,j,j}\,.
\]
A direct computation gives $E_1v_1^i = d\,v_1^i$, $E_1v_2^i = v_1^i$, $E_2v_1^i = v_2^i$ and $E_2v_2^i = d\,v_2^i$, so that $Z = \mathrm{span}\{v_1^i,v_2^i : i=1,...,d\}$ is invariant when acted on by the Hamiltonian. Moreover $\langle v_1^i,v_1^l\rangle = \langle v_2^i,v_2^l\rangle = d\,\delta_{il}$ and $\langle v_1^i,v_2^l\rangle = \delta_{il}$, so $Z$ splits into $d$ different mutually orthogonal two-dimensional sectors, and $\dim Z = 2d$.

Since $E_1$ has image $\ket{\phi} \otimes \cc^d \cong \mathrm{span}\{v_1^i : i=1,...,d\}$ and $E_2$ has image $\cc^d \otimes \ket{\phi} \cong \mathrm{span}\{v_2^i : i=1,...,d\}$, we have $\Im(E_1+E_2)\subseteq Z$. Within the $i$-th sector, for each $i=1,\dots,d$, the matrix of $E_1+E_2$ in the basis $(v_1^i,v_2^i)$ is
\[
    \begin{pmatrix} d & 1\\ 1 & d\end{pmatrix}
\]
with eigenvalues $d\pm1$, both nonzero as $d \geq 3$. Hence $E_1+E_2$ restricts to an invertible operator on $Z$, so $(E_1+E_2)(Z) = Z$. In particular, $Z \subseteq \Im(E_1+E_2)$ so $\Im(E_1+E_2) = Z$. Being self adjoint, $\ker(E_1+E_2) = Z^{\perp}$, of dimension $d^3-2d$. Dividing the eigenvalues by $d$ gives \eqref{eq:three_site_spectrum}.
\end{proof}

\begin{cor}\label{cor:knabe_bound}
For every $L \geq 6$, the spectral gap of $H_L^{\text{PBC}}$ is bounded below by $1-\tfrac2d$.
\end{cor}
\begin{proof}
We apply Theorem~\ref{thm:knabe} with $m=2$, which is admissible for $\lfloor L/2\rfloor > 2$, that is for $L \geq 6$. By Proposition~\ref{prop:three_site}, $\epsilon_2 = 1-\tfrac1d > \tfrac12$ given $d \geq 3$, so that
\[
    \eta_2 = \frac{2}{1}\Big(\epsilon_2 - \frac12\Big) = 2\Big(\tfrac12-\tfrac1d\Big) = 1-\frac2d\,. \qedhere
\]
\end{proof}

We now have all the necessary results to prove Theorem~\ref{thm:gns_gap}.

\begin{proof}[Proof of Theorem~\ref{thm:gns_gap}]
Recall the strategy was to derive an upper and lower bound for the GNS gap, and see they agree. For the upper bound, Proposition~\ref{prop:essential_spectrum} showed $1-\tfrac2d \in \sigma(H)$. Since $H\Omega = 0$ and $H \geq 0$, the gap is at most $1-\tfrac2d$.

For the lower bound, Corollary~\ref{cor:knabe_bound} gives $\sigma\big(H^{\Phi_n}_{\Lambda_n}\big)\cap\big(0,1-\tfrac2d\big) = \emptyset$ for every $n$ with $N_n \geq 6$, since $H^{\Phi_n}_{\Lambda_n} = H^{\text{PBC}}_{N_n}$ is the periodic Hamiltonian on $N_n$ sites. Hence the hypothesis of Proposition~\ref{prop:spectrum_transfer} is satisfied for every $E \in (0,1-\tfrac2d)$, and therefore $\sigma(H)\cap(0,1-\tfrac2d) = \emptyset$. That is, the spectral gap of $H$ is at least $1-\tfrac2d$.
\end{proof}

\subsection{The Maximally Mixed Ground State}
\label{sec:other_ground_states}

Theorem~\ref{thm:gns_gap} concerns the fully polarized state $\omega$, which on finite chains is one ground state among the exponentially many counted by Theorem~\ref{thm:jw}. We now construct a second infinite volume ground state, obtained by averaging over the whole ground state space rather than selecting a distinguished vector in it, and show that it too is gapped. The construction rests on the fact that the partial trace of the Jones--Wenzl projection is again a Jones--Wenzl projection. Recall we let $P_L$ denote the ground state projector on $[1,L]$, and it satisfies the Jones--Wenzl recursion in Theorem~\ref{thm:jw}.

\begin{lem}\label{lem:partial_trace_JW}
For every $L \geq 2$,
\begin{equation}\label{eq:partial_trace_JW}
    \Tr_L\big(P_L\big) = \frac{D_L}{D_{L-1}}\,P_{L-1}\,,
\end{equation}
where $\Tr_L$ denotes the partial trace over the site $L$.
\end{lem}
\begin{proof}
$P_{L-1}$ acts as identity on site $L$ so that $\Tr_L(P_{L-1}) = d\,P_{L-1}$. Additionally, $\Tr_L(E_{L-1}) = \ds_{[1,L-1]}$ so that
\[
    \Tr_L\big(P_{L-1}E_{L-1}P_{L-1}\big) = P_{L-1}\,\Tr_L\big(E_{L-1}\big)\,P_{L-1} = P_{L-1}\,.
\]
Then by Theorem~\ref{thm:jw}
\[\Tr_L(P_L) = \Tr_L(P_{L-1}) - \frac{D_{L-2}}{D_{L-1}} \Tr_L(P_{L-1}E_{L-1}P_{L-1}) = \left(d-\frac{D_{L-2}}{D_{L-1}}\right)P_{L-1} = \frac{D_L}{D_{L-1}}P_{L-1}\,.\]
\end{proof}

\begin{definition}\label{def:tracial}
Let $\Lambda \subset \zz$ be an interval of size $m$, and let $\bar\omega_\Lambda$ be the state on $\cA_{\Lambda}$ with density matrix $P_m/D_m$, that is
\begin{equation}\label{eq:tracial_state}
    \bar\omega_\Lambda(A) = \frac{1}{D_m}\Tr\big(P_m\, A\big)\,, \qquad A \in \cA_{\Lambda}\,,
\end{equation}
where $P_m$ is understood to act on $\cA_\Lambda$ through the identification of $\Lambda$ with $[1,m]$.
\end{definition}

Thus $\bar\omega_\Lambda$ is the state of maximal entropy on the ground state space of $H_\Lambda$; equivalently it is the $T \to 0$ limit of the Gibbs states of $H_\Lambda$.

\begin{prop}\label{prop:tracial_consistent}
The family $\{\bar\omega_\Lambda\}$, indexed by finite intervals $\Lambda \subset \zz$, is consistent. It therefore extends uniquely to a state $\bar\omega$ on $\cA$, which is translation invariant and satisfies
\begin{equation}\label{eq:tracial_restriction}
    \bar\omega\big|_{\cA_\Lambda} = \frac{1}{D_m}\,\Tr\big(P_m\,\cdot\,\big)
\end{equation}
for every interval $\Lambda$ of length $m$.
\end{prop}
\begin{proof}
Let $\Lambda' \subseteq \Lambda$ be intervals, of lengths $m' \leq m$. Since $\Lambda\setminus\Lambda'$ consists of a segment to the left of $\Lambda'$ and a segment to its right, it suffices to treat the two cases where a single site is removed from one end of $\Lambda$; the general case follows by iterating. Throughout we identify $\Lambda$ with $[1,m]$. Removing the rightmost site is Lemma~\ref{lem:partial_trace_JW}: tracing out the site $m$ in \eqref{eq:tracial_state} gives the density matrix $\Tr_m(P_m)/D_m = P_{m-1}/D_{m-1}$, which is that of $\bar\omega_{\Lambda\setminus\{m\}}$.

Removing the leftmost site follows from this together with the reversal symmetry of the model. Let $R$ be the unitary on $\cH_{[1,m]}$ determined by $R\big(v_1\otimes\cdots\otimes v_m\big) = v_m\otimes\cdots\otimes v_1$. Then $RP^{(0)}_{x,x+1}R^{*} = P^{(0)}_{m-x,m-x+1}$, so $R$ permutes the interaction terms of $H_m$ among themselves, whence $RH_mR^* = H_m$, $R$ preserves $\ker H_m$, and $RP_mR^* = P_m$. Write $R':\cH_{[2,m]} \to \cH_{[1,m-1]}$ for the analogous reversal on $[2,...,m]$, followed by the translation identifying $[2,m]$ with $[1,m-1]$. That is, $R\big(\ket{j}_1\otimes v\big) = R'v\otimes\ket{j}_m$ for every $v \in \cH_{[2,m]}$. Then for $v, w \in \cH_{[2,m]}$,
\begin{align*}
    \big\langle v,\ \Tr_1(P_m)\,w\big\rangle
    &= \sum_{j}\big\langle \ket{j}_1\otimes v,\ P_m\,\ket{j}_1\otimes w\big\rangle\\
    &= \sum_{j}\big\langle R\big(\ket{j}_1\otimes v\big),\ RP_mR^*\,R\big(\ket{j}_1\otimes w\big)\big\rangle\\
    &= \sum_{j}\big\langle R'v\otimes\ket{j}_m,\ P_m\,\left(R'w\otimes\ket{j}_m\right)\big\rangle
     = \big\langle R'v,\ \Tr_m(P_m)\,R'w\big\rangle\\
    &= \frac{D_m}{D_{m-1}}\big\langle R'v,\ P_{m-1}R'w\big\rangle
     = \frac{D_m}{D_{m-1}}\big\langle v,\ R'^*P_{m-1}R'\,w\big\rangle\,.
\end{align*}
By definition of $R'$ and reversal invariance of $P_{m-1}$ on $[1,m-1]$, we have $R'^*P_{m-1}R'$ is the ground state projection $P_{m-1}$ now understood on $[2,m]$. Tracing out the first site in \eqref{eq:tracial_state} thus gives the density matrix of the restriction of $\bar\omega_\Lambda$ on the remaining sites, which shows consistency. 

A consistent family of states on the local algebras of the intervals determines a unique state on $\cA_{\text{loc}}$: for finite $X \subset \zz$ choose an interval $\Lambda \supseteq X$ and restrict $\bar\omega_\Lambda$ to $\cA_X$, which is independent of the choice by consistency. This extends by density to a state $\bar\omega$ on $\cA$, whose restriction to $\cA_\Lambda$ is $\bar\omega_\Lambda$ for every interval $\Lambda$, which is \eqref{eq:tracial_restriction}.

For translation invariance, let $\alpha_y$ denote the shift automorphism of $\cA$, so that $\alpha_y$ carries $\cA_\Lambda$ onto $\cA_{\Lambda+y}$. The identification of any length $m$ interval with $[1,m]$ as in Definition~\ref{def:tracial} gives that the state acts with the same density matrix on any such interval. Thus $\bar\omega_{\Lambda+y}\circ\alpha_y = \bar\omega_\Lambda$ for every interval $\Lambda$. Hence $\bar\omega\circ\alpha_y$ and $\bar\omega$ agree on every $\cA_\Lambda$, and therefore on $\cA$.
\end{proof}

The following proposition bounds the GNS gap of $\bar\omega$ from below by $1-\tfrac2d$. \cite{bachmann2016lieb} defines infinite volume ground states as weak-$*$ limits of pure finite volume ground states, but nothing in the proof of their Proposition 5.4 requires this. This is why we then apply their result to the case of $\bar\omega$, which is given as a weak-$*$ limit point of mixed ground states. We do not argue an upper bound of $1-\tfrac2d$ to obtain the exact gap as we did for the GNS Hamiltonian of the fully polarized state, as we have not yet found the right trial states.

\begin{prop}\label{prop:tracial_gapped}
$\bar\omega$ is a ground state of $(\cA, \tau_t)$, and its GNS Hamiltonian $\bar H$ has spectral gap at least $1 - \tfrac{2}{d}$.
\end{prop}
\begin{proof}
Take $A \in \cA_{\text{loc}}$ and $\Lambda$ an interval of length $m$ containing $\mathrm{supp}(A)$ along with its left and right neighbors. Then
\begin{align*}
\bar\omega\left(A^*\delta(A)\right) &= \sum_{X:\,X \cap \mathrm{supp}(A) \neq \emptyset} \bar\omega\left(A^*[\Phi(X),A]\right)\\
&= \sum_{X:\,X \cap \mathrm{supp}(A) \neq \emptyset} \frac{1}{D_m} \Tr\left(P_m A^*\Phi(X)A - P_m A^*A\Phi(X)\right)\\
&= \sum_{X:\,X \cap \mathrm{supp}(A) \neq \emptyset} \frac{1}{D_m} \Tr\left(P_m A^*\Phi(X)A\right) \\
&= \sum_{X:\,X \cap \mathrm{supp}(A) \neq \emptyset} \frac{1}{D_m} \Tr\left(P_m (\Phi(X)A)^*(\Phi(X)A)\right) \geq 0
\end{align*}
where the third equality is by cyclicity of trace and using $\Phi(X)P_m = 0 = P_m\Phi(X)$ for each possible $X$, the fourth equality was since $\Phi(X)$ is either $0$ or a projection, and the result is nonnegative given $P_m (\Phi(X)A)^*(\Phi(X)A)P_m \geq 0$. Hence $\bar\omega$ is a ground state.

Now we can let $\bar H$ be the GNS Hamiltonian of $\bar\omega$ with respect to $\tau_t$, and lower bound the gap. Let $\Gamma$ be an arbitrary finite subset of $\zz$ and decompose $\Gamma = I_1\sqcup\cdots\sqcup I_r$ as a disjoint union of maximal intervals. If each interval in the decomposition has length one, $H_\Gamma=0$. A bond $\{x,x+1\}$ is contained in $\Gamma$ only if both endpoints lie in the same $I_j$, so $H_\Gamma = \sum_j H_{I_j}$ is a sum of commuting operators acting on disjoint sets of sites, and
\[
    \sigma(H_\Gamma) = \Big\{\sum_{j} \xi_j \ :\ \xi_j \in \sigma\big(H_{I_j}\big)\Big\}\,.
\]
A nonzero element of this set has $\xi_j > 0$ for some $j$, as zero is the minimum energy on each $I_j$ by frustration freeness, and so is at least $\min\{\gamma_{|I_j|} : j \text{ with } |I_j|\geq 2\}$. By Corollary~\ref{cor:fv_gap}, $\gamma_{|I_j|} = 1-\tfrac2d\cos\left(\pi/|I_j|\right)$ is the gap of $H_{I_j}$ for $|I_j| \geq 2$. Since $\gamma_{|I_j|} > 1-\tfrac2d$ for every $|I_j|$, the spectrum of $H_\Gamma$ misses $\big(0,\,1-\tfrac2d\big)$ for every finite $\Gamma$.

Now let $E \in \big(0,1-\tfrac2d\big)$ and choose $\epsilon>0$ with $(E-\epsilon,E+\epsilon)\subset\big(0,1-\tfrac2d\big)$. The hypothesis of Proposition 5.4 of \cite{bachmann2016lieb} is then satisfied so that $E \notin \sigma(\bar H)$. As $E$ was arbitrary, $\sigma(\bar H)\cap\big(0,1-\tfrac2d\big) = \emptyset$.
\end{proof}

By \eqref{eq:tracial_restriction} the restriction of $\bar\omega$ to an interval of length $L$ is maximally mixed on a space of dimension $D_L$, so its entropy is $\ln D_L$ exactly, so the mean entropy is given by
\begin{equation}\label{eq:tracial_entropy_density}
    \lim_{L \to \infty}\frac{1}{L}S\big(\bar\omega|_{\cA_{[1,L]}}\big) = \lim_{L \to \infty}\frac{1}{L} \ln \frac{q^{L+1} - q^{-(L+1)}}{q-q^{-1}} = \ln q > 0
\end{equation}
given $q > 1$, whereas $\omega$ is a product state and has vanishing entropy in every volume. 

We also consider the mutual information between two adjacent blocks $A = [1,k]$ and $C = [k+1,L]$, so that $A \cup C = [1,L]$. By Proposition~\ref{prop:tracial_consistent} the restriction of $\bar\omega$ to each of the three regions is of the form $P_m/D_m$ with $m$ the length of the corresponding interval, so the entropy of each restricted state is maximal and thus a logarithm of $D_m$:
\begin{equation}\label{eq:tracial_mutual_info}
\begin{aligned}
    I(A\!:\!C) &= \ln D_k + \ln D_{L-k} - \ln D_L = \ln\frac{D_kD_{L-k}}{D_L} \\ 
    &= \ln\frac{q}{q-q^{-1}} + \ln\frac{\big(1-q^{-2(k+1)}\big)\big(1-q^{-2(L-k+1)}\big)}{1-q^{-2(L+1)}}\,.
\end{aligned}
\end{equation}
Thus
\begin{equation}\label{eq:tracial_mutual_info_limit}
    I(A\!:\!C) \;\longrightarrow\; \ln\frac{q}{q-q^{-1}}
    \qquad\text{as } \min(k,L-k)\to\infty\,.
\end{equation}
Therefore $\bar\omega$ obeys an area law for mutual information while carrying a volume law for entropy.

\section{The Finite Volume Ground State Projections}
\label{sec:projectors}

In this section we prove Theorem~\ref{thm:jw}: the ground state projectors on open chains satisfy the Jones--Wenzl recursion, and consequently so do the degeneracies.


\subsection{An Algebraic Characterization of Ground State Projections}
\label{sec:algebraic}

We first isolate the general facts about frustration free models that characterize ground state projectors algebraically. Throughout this subsection $\Xi$ is an arbitrary frustration free interaction with $\Xi(X)\geq 0$. For $\Lambda$ a finite subset of the lattice we let $H_\Lambda^\Xi = \sum_{X\subseteq\Lambda}\Xi(X)$ be the corresponding finite volume Hamiltonian, $P_\Lambda$ the orthogonal projection onto $\ker H_\Lambda^\Xi$, and
\begin{equation}\label{eq:interaction_algebra}
     I_\Lambda \coloneq \text{ the non-unital $\cc$-algebra generated by } \{\Xi(X)\}_{X\subseteq\Lambda}\,,
\end{equation}
Frustration freeness says $P_\Lambda \Xi(X) = \Xi(X)P_\Lambda = 0$ for every $X \subseteq \Lambda$, and it implies the nesting relations $P_\Lambda P_{\Lambda'} = P_{\Lambda'}P_\Lambda = P_{\Lambda'}$ for $\Lambda \subseteq \Lambda'$.

\begin{prop}\label{prop:algebra_condition}
For any frustration free model, $\ds - P_\Lambda \in I_\Lambda$.
\end{prop}
\begin{proof}
The operator $\ds - P_\Lambda$ is the orthogonal projection onto $\Im H_\Lambda$. Let $\lambda_1,\dots,\lambda_r$ be the nonzero eigenvalues of $H_\Lambda^\Xi$ and let $T$ be an interpolating polynomial with $T(0)=0$ and $T(\lambda_i)=1$; it has no constant term. By the functional calculus $\ds - P_\Lambda = T(H_\Lambda^\Xi)$, and since $H_\Lambda^\Xi$ is a linear combination of the generators of $I_\Lambda$, any polynomial in $H_\Lambda$ without constant term lies in $I_\Lambda$.
\end{proof}

\begin{prop}\label{prop:uniqueness}
Let $Q$ be an operator on $\cH_\Lambda$ with $Q\Xi(X) = \Xi(X)Q = 0$ for all $X \subseteq \Lambda$ and $\ds - Q \in I_\Lambda$. Then $Q = P_\Lambda$. In particular $Q$ is the orthogonal projection onto the ground state space of $H_\Lambda$.
\end{prop}
\begin{proof}
By hypothesis, $Q = \ds - F_Q$ for some $F_Q \in I_\Lambda$. Since $F_Q$ is a polynomial in the generators of $I_\Lambda$ with no constant term, $QF_Q = F_QQ = 0$. Similarly for $F_P \coloneq \ds - P_\Lambda \in I_\Lambda$ we have $P_\Lambda F_Q = F_QP_\Lambda = QF_P = F_PQ = 0$. Therefore
\[
    Q = (\ds - F_P)Q = P_\Lambda Q = P_\Lambda(\ds - F_Q) = P_\Lambda\,. \qedhere
\]
\end{proof}

The first condition shows $Q$ is supported in the ground state space, and the second condition shows it acts as identity there. Proposition~\ref{prop:uniqueness} is what allows the recursion below to be established by an ansatz. It suffices to exhibit \emph{some} operator built from the interaction terms satisfying the conditions, and it is then automatically the ground state projection, with idempotency and orthogonality coming for free.


\subsection{The Ground State Projectors and Degeneracies}
\label{sec:jw}

We now return to the singlet model on the open chain, where the algebra $I_{[1,L]}$ from the previous section is the non-unital algebra generated by $\{E_x\}_{x=1}^{L-1}$. This is not the image of $TL_L(d)$ under $\rho_d$, which is unital. 

\begin{proof}[Proof of Theorem~\ref{thm:jw}]
We consider the operator
\begin{equation}\label{eq:ansatz}
    G_L = P_{L-1} - c_{L-1}\,P_{L-1}E_{L-1}P_{L-1}
\end{equation}
and determine $c_{L-1}$ as needed from the requirement that $G_L$ satisfies the algebraic conditions of Proposition~\ref{prop:uniqueness} to obtain $G_L = P_L$. For the base case $L=2$ we see 
\[P_2 = \ds - \tfrac1d E_1 = P_1 - \frac{\Omega_0}{\Omega_1} P_1 E_1 P_1\,.\]
We also do the $L=3$ case as it is instructive. The Temperley--Lieb relations give
\begin{equation}\label{eq:E2_P2_E2}
    E_2P_2E_2 = E_2^2 - \tfrac1d E_2E_1E_2 = dE_2 - \tfrac1d E_2 = \frac{d^2-1}{d}\,E_2 = \frac{\Omega_2}{\Omega_1}E_2\,.
\end{equation}
Imposing $E_2G_3 = 0$, or equivalently $G_3E_2=0$, on \eqref{eq:ansatz} forces $c_2 = \Omega_1/\Omega_2$ since $E_2P_2 \neq 0$ (as seen by applying it to the vector $\ket{s,s,-s}$). We also see by definition of $G_3$ that $E_1G_3=0=G_3E_1$, and $\ds-G_3$ is in the non-unital algebra generated by $\{E_1,E_2\}$, again using $P_2 = \ds - \tfrac1d E_1$. Thus $G_3 = P_3$  by Proposition~\ref{prop:uniqueness}. 

For the inductive step assume the ground state projector on the first $L-1$ sites of the chain $[1,L]$ is given by $P_{L-1} = P_{L-2}-c_{L-2}P_{L-2}E_{L-2}P_{L-2}$ with $c_{L-2} = \Omega_{L-3}/\Omega_{L-2}$. By the Temperley-Lieb relations and that $P_{L-2}$ commutes with $E_{L-1}$,
\begin{align}\label{eq:EL_PL_EL}
    E_{L-1}P_{L-1}E_{L-1} &= E_{L-1}P_{L-2}E_{L-1} - c_{L-2}P_{L-2}E_{L-1}E_{L-2}E_{L-1}P_{L-2} \nonumber\\
    &= P_{L-2}E_{L-1}^2 - c_{L-2}P_{L-2}E_{L-1}P_{L-2} = \big(d - c_{L-2}\big)P_{L-2}E_{L-1}\,.
\end{align}
Imposing $E_{L-1}G_L=0$, or equivalently $G_LE_{L-1}=0$, on \eqref{eq:ansatz} and using the nesting relation $P_{L-2}P_{L-1}=P_{L-1}$ gives
\begin{equation}\label{eq:EL_PL_constraint}
    E_{L-1}P_{L-1} = c_{L-1}\big(d-c_{L-2}\big)E_{L-1}P_{L-1}\,, \qquad \text{ hence } \qquad c_{L-1} = \frac{1}{d-c_{L-2}}\,.
\end{equation}
Here we similarly used as in the $L=3$ case that $E_{L-1}P_{L-1} \neq 0$, as seen by applying it to the vector $\ket{s,\dots,s, -s}$. Solving with $c_k = \Omega_{k-1}/\Omega_k$ converts this continued fraction into the desired recursion $\Omega_{L-1} = d\,\Omega_{L-2}-\Omega_{L-3}$. It remains to identify the operator $G_L$ we constructed with $P_L$. By the inductive hypothesis $P_{L-1}$ annihilates $E_x$ for $x \leq L-2$, and this is inherited by \eqref{eq:ansatz}; together with $E_{L-1}G_L = G_LE_{L-1}=0$, the operator $G_L$ annihilates every interaction term. Moreover
\[
    \ds - G_L = \big(\ds - P_{L-1}\big) + c_{L-1}P_{L-1}E_{L-1}P_{L-1}\,,
\]
with each summand lying in $I_{[1,L]}$ by the inductive hypothesis. Thus $G_L = P_L$ by Proposition~\ref{prop:uniqueness}. 

Now we move on to the second part of the theorem, where by taking traces in \eqref{eq:jw_recursion} we convert the operator identity into a scalar recursion for the ground state degeneracies. We have $D_2 = \Tr P_2 = d^2-1 = \Omega_2$. Assume $D_k = \Omega_k$ for $k<L$ and take the trace of \eqref{eq:jw_recursion} over $\cH_L$. Using cyclicity, idempotency of $P_{L-1}$, and the partial trace identity $\Tr_L(E_{L-1}) = \ds_{[1,L-1]}$,
\begin{align}\label{eq:DL_trace}
    D_L &= \Tr(P_{L-1}) - c_{L-1}\Tr_{[1,L-1]}\big(P_{L-1}\Tr_L(E_{L-1})\big) \nonumber\\
    &= d\,D_{L-1} - c_{L-1}D_{L-1}\,.
\end{align}
Substituting $c_{L-1} = \Omega_{L-2}/\Omega_{L-1} = D_{L-2}/D_{L-1}$ gives $D_L = d\,D_{L-1}-D_{L-2}$. Since $D_L$ and $\Omega_L$ obey the same recursion with the same initial data, they agree.

It remains to solve the recursion. The characteristic equation is $\phi^2 - d\phi+1=0$, with roots
\begin{equation}\label{eq:char_roots}
    \phi_\pm = \frac{d\pm\sqrt{d^2-4}}{2}\,, \qquad \phi_+\phi_- = 1\,,
\end{equation}
so we may write $\phi_- = \phi_+^{-1}$ and $D_L = A\phi_+^{L}+B\phi_+^{-L}$. The initial conditions $\Omega_0=1$ and $\Omega_1 = d$ give $A+B=1$ and $A\phi_+ + B\phi_+^{-1}=d$; using $d = \phi_++\phi_+^{-1}$, the second reads $A(\phi_+-\phi_+^{-1}) + \phi_+^{-1} = \phi_++\phi_+^{-1}$, so $A = \phi_+/(\phi_+-\phi_+^{-1})$ and $B = -\phi_+^{-1}/(\phi_+-\phi_+^{-1})$. Hence
\begin{equation}\label{eq:DL_explicit}
    D_L = \frac{\phi_+^{L+1}-\phi_+^{-(L+1)}}{\phi_+-\phi_+^{-1}}\,.
\end{equation}
Rewriting $q \coloneq \phi_+$, the defining relation $\phi_++\phi_+^{-1}=d$ becomes $q+q^{-1}=d$, which is the standard parametrization of the Temperley--Lieb loop weight, and \eqref{eq:DL_explicit} is precisely
the quantum integer
\begin{equation}\label{eq:quantum_integer}
    D_L = [L+1]_q = \frac{q^{L+1}-q^{-(L+1)}}{q-q^{-1}}\,.
\end{equation}
\end{proof}

Under the identification of the singlet model with the spin-$\tfrac12$ $XXZ$ chain at anisotropy $\Delta = d/2$, this is the familiar $q$-deformation of the dimension $L+1$ maximal spin multiplet: as $q \to 1$, that is as $d \to 2$ and $s \to \tfrac12$, the right hand side of \eqref{eq:quantum_integer} tends to $L+1$. For $s \geq 1$ one has $q > 1$, and \eqref{eq:quantum_integer} grows exponentially. In the spin-$1$ case $d=3$, so that $q = \phi_g^{2}$ with $\phi_g = \tfrac{1+\sqrt5}{2}$ being the golden ratio and $q-q^{-1} = \sqrt5$. Then \eqref{eq:quantum_integer} reads
\begin{equation}\label{eq:spin1_degeneracy}
    D_L = \frac{\phi_g^{\,2L+2}-\phi_g^{-(2L+2)}}{\sqrt5} = F_{2L+2}\,,
\end{equation}
using Binet's formula for the $n$-th Fibonacci number $F_n$. That is, the spin-$1$ ground state degeneracies are the even-indexed Fibonacci numbers. 

\section{Discussion}
\label{sec:discussion}

The two features of the singlet model established above, the strictly positive gap and the exponentially degenerate ground state space, are not independent. This section discusses two ways in which their interplay distinguishes the model from other chains with ferromagnetic ground states: it breaks a continuous symmetry without producing gapless excitations, and its gap is unstable.


\subsection{Symmetry Breaking Without Gapless Goldstone Modes}
\label{sec:goldstone}

The fully polarized state $\omega$ of \eqref{eq:omega} breaks the global continuous $SU(2)$ symmetry of the model, so one expects such breaking to be accompanied by gapless excitations in the spectrum of the GNS Hamiltonian $H$. Theorem~\ref{thm:gns_gap} shows that it is not. Instead, infinite degeneracy in the ground state space takes its place. Let $A_x \in \cA_{\text{loc}}$ be the local operator lowering the spin at the site $x$, so that $A_x\ket{s\cdots s} = \ket{s-1}_x$. Then for every $x \in \zz$ the vector $\pi(A_x)\Omega$ is orthogonal to $\Omega$:
\begin{equation}\label{eq:gns_ortho}
    \langle \Omega, \pi(A_x)\Omega\rangle = \bra{s\cdots s}A_x\ket{s\cdots s} = \braket{s\cdots s}{s-1}_x = 0\,.
\end{equation}
These vectors are also all mutually orthogonal given $\langle \pi(A_x)\Omega, \pi(A_y)\Omega\rangle = 0$ when $x \neq y$. We also have $H\pi(A_x)\Omega = 0$. Indeed, recall $\omega(A^*\delta(A)) = \langle \pi(A)\Omega, H\pi(A)\Omega\rangle$. Also, the states $\ket{s}\otimes\ket{s-1}$ and $\ket{s-1}\otimes\ket{s}$ have no singlet component. Evaluating in a finite volume gives
\begin{align}\label{eq:gns_energy}
    \langle \pi(A_x)\Omega, H\pi(A_x)\Omega\rangle &= \omega\big(A_x^*\delta(A_x)\big) = \bra{s\cdots s}A_x^*[H_L,A_x]\ket{s\cdots s} \nonumber\\
    &= \bra{s-1}_x H_L\ket{s-1}_x - \bra{s-1}_x H_L\ket{s\cdots s} = 0\,.
\end{align}
In particular $\ker H$ is infinite dimensional.

We compare this situation with the following specialized formulation of Goldstone's theorem stated in contrapositive, which holds for more general sufficiently decaying interactions as in \cite{landau1981energy}.

\begin{thm}\label{thm:goldstone}
Let $\tau_t$ be the infinite volume dynamics as described in Section~\ref{sec:cstar}. Let $\{\alpha_u : \cA \to \cA\}_{u\in\rr}$ be a continuous symmetry of the dynamics. That is, $\alpha_u\circ\tau_t = \tau_t\circ\alpha_u$ where
\[
    \alpha_u(A) = e^{iu\sum_{x\in\Lambda}J_x}\,A\,e^{-iu\sum_{x\in\Lambda}J_x} \quad \text{ for each } A \in \cA_\Lambda\,,
\]
with $J_x = J_x^* \in \cB(V_s)$ independent of $x$. Suppose $\omega_g$ is a translation invariant ground state of $(\cA,\tau_t)$ for which there exists $\gamma>0$ with
\begin{equation}\label{eq:goldstone_hypothesis}
    \gamma\,\omega_g(A^*A) \leq \omega_g\big(A^*\delta(A)\big) \qquad \text{ for all } A \in \cA_{\text{loc}} \text{ with } \omega_g(A)=0\,.
\end{equation}
Then $\omega_g$ is invariant under the symmetry: $\omega_g\circ\alpha_u = \omega_g$ for all $u \in \rr$.
\end{thm}

The hypothesis \eqref{eq:goldstone_hypothesis} is strictly stronger than the GNS Hamiltonian $H_g$ of $\omega_g$ being gapped. It asserts simultaneously that $\omega_g$ is a ground state of the dynamics $\tau_t$, that $\ker H_g$ is one dimensional, and that $H_g$ has a gap bounded below by $\gamma$. Our arguments show that the singlet model satisfies the first and third of these but not the second: the vectors $\pi(A_x)\Omega$ have zero energy and vanishing expectation in $\omega$, so \eqref{eq:goldstone_hypothesis} fails for $A = A_x$ despite the gap. As stated earlier, the symmetry breaking of $\omega$ is therefore accompanied not by gapless excitations but by an infinite degeneracy at exactly zero energy. This gives the singlet model its particular manifestation of the Goldstone theorem. 


\subsection{Instability of the Gap}
\label{sec:instability}

We argue the gap of the singlet model is unstable, in the sense that there exist arbitrarily small $SU(2)$-invariant perturbations for which the thermodynamic gap closes. Let us first introduce the class of models we perturb with, which were briefly discussed in the introduction for the $s=1$ case. A translation invariant, rotation invariant, nearest neighbor interaction $h_{x,x+1}$ on a spin-$s$ chain may be written, up to a shift of the spectrum, either as a polynomial in $\vec{S}\cdot\vec{S}$ or in terms of projections onto the possible spin subspaces of two neighboring spins:
\begin{equation}\label{eq:general_interaction}
    h_{x,x+1} = \sum_{j=1}^{2s}J_j\big(\vec{S}_x\cdot\vec{S}_{x+1}\big)^j = \sum_{j=0}^{2s}K_jP^{(j)}_{x,x+1}\,.
\end{equation}
Hence there is a phase diagram structure within the parameter space $\{K_0,\dots,K_{2s}\}$. What follows is an attempt to transport that structure from $s=1$, where the phase diagram is better understood, to higher spin: we write down inequalities among the $K_j$ that at $s=1$ cut out the known phases, observe that they make sense for every $s$, and take the regions they define at $s=\tfrac32$ as our working guess for where the same behavior occurs. This is the basis for our expectation that the regions into which we perturb the singlet model below are gapless. We stress that it is an expectation and not a theorem. 

The ferromagnetic models are those with $K_{2s} < K_0,\dots,K_{2s-1}$, where the ground state favors nearest neighbor spins being in the maximal spin subspace. For example, the Heisenberg ferromagnet with interaction $-\vec{S}_x \cdot \vec{S}_{x+1}$ is included in this family, and has $\ket{s\cdots s}$ as a ground state. These models have ground state space given by the maximal spin subspace of the chain, of dimension $2sL+1$, and carry gapless spin wave excitations
\begin{equation}\label{eq:su2_magnon}
    \psi_k = \sum_{x=1}^{L}e^{ikx}\ket{s-1}_x\,,
\end{equation}
where recall $\ket{s-1}_x$ was our notation for the state with spin $s$ at each site except for $x$ which has spin $s-1$. The energies of these states are computed in Proposition~\ref{thm:spinwave_energies} for all members of the family in \ref{eq:general_interaction}. Equality of $K_{2s}$ with at least one of the larger coefficients places the model on the boundary of the ferromagnetic phase, for example, the singlet model corresponds to $K_{2s} = K_{2s-1} = \cdots = K_1 = 0$ and $K_0 = 1$. 

A second region is relevant to our purpose. In the $s=1$ case, after normalizing to rewrite the expression in Equation~\eqref{eq:general_interaction} as in Equation~\eqref{eq:spin1_family}, the region is given by $\theta \in (\tfrac\pi4, \tfrac\pi2)$, lying between the Uimin--Lai--Sutherland (ULS) model $\theta=\tfrac\pi4$ and the purely biquadratic point $\theta=\tfrac\pi2$, which up to a shift and normalization is the singlet model. As mentioned in the introduction, this phase is conjectured to be gapless with soft modes at momenta $\pm \tfrac{2\pi}{3}$, and critical (power law decay of correlations). In terms of projection operators the region is given by $K_0 > K_2 > K_1$. Rewritten as
\begin{equation}\label{eq:critical_spin1}
    K_0,\,K_2 \;>\; \min_{j} K_j\,, \qquad K_2 \;<\; \max_{j} K_j\,,
\end{equation}
these conditions can be generalized to every $s$. We therefore consider the two nested regions
\begin{align}\label{eq:regions_A&B}
    A_s \coloneq \Big\{\, K_0,\,K_{2s} > \min_{j} K_j \,\Big\}\,, 
    \qquad 
    B_s \coloneq A_s \cap \Big\{\, K_{2s} < \max_{j} K_j \,\Big\}
    \;\subseteq\; A_s\,.
\end{align}
At $s=1$ these regions reproduce known ones: parametrizing as in Equation~\eqref{eq:spin1_family}, $A_1 = \Big(\arctan\tfrac13,\ \tfrac{\pi}{2}\Big)$ and $B_1 = \Big(\tfrac{\pi}{4},\ \tfrac{\pi}{2}\Big)$. The region $A_1$ has on its lower endpoint the AKLT model, where real space correlations first become incommensurate with the lattice \cite{schollwock1996}. At the point $\theta \cong  0.1314\pi$ within $A_1$, the structure factor (defined as the Fourier transform of the real space correlation function) has peaks that begin to drift away from $\pi$ towards $\tfrac{2\pi}{3}$ when $\theta$ reaches $\tfrac\pi4$ \cite{bursill1995}, in agreement with the soft mode found at momenta $\tfrac{2\pi}{3}$ for the ULS model. Soft modes at this value of momenta are expected to persist throughout the critical region $B_1$ \cite{fath1991period}.

We expect a similar situation in the $s=\tfrac32$ case. Here, the boundary of $A_{3/2}$ passes through the ULS model and through $P^{(3)}$, the maximal spin projector and so the analogue of the AKLT model; this matches the $s=1$ pattern. Preliminary DMRG computations near $P^{(3)}$, and at $\left(\vec{S}\cdot\vec{S}\right)^2$, to be reported elsewhere, find the structure factor peaked away from $\pi$ inside $A_{3/2}$ and $B_{3/2}$, and peaked at $\pi$ outside $A_{3/2}$, again as in $s=1$. The similarity to the $s=1$ case for the behavior of the structure factor, the structure of the inequalities, and the models lying on the boundaries of the regions they define, are our reasons for believing $B_{3/2}$ is the $s=\tfrac32$ analog of the critical phase found in the $s=1$ case. This is further supported by the plots of the spectra on finite length chains in Figure~\ref{fig:H_eps,delt_specs}.

Comparing the inequalities that describe these two gapless phases, we see the singlet model sits between them. We find evidence below that its gap does not survive perturbation in either direction. The mechanism of instability is partially visible at the level of dimensions. Recall the ground state degeneracy of the singlet model, $D_L = [L+1]_q$, grows exponentially, whereas that of the neighboring ferromagnets grows only linearly. Thus, switching off a perturbation from the interior of the ferromagnetic region to the singlet model, the $D_L$ lowest eigenvalues of the perturbed model must collapse onto zero. For finite $L$ this collapse is continuous in the perturbation parameter, but the number of levels involved grows exponentially with $L$, and one expects the closing of the gap to become abrupt in the thermodynamic limit. Some of these excited states on the ferromagnetic side that collapse to zero are the spin waves \eqref{eq:su2_magnon} and illustrate the collapse directly: the explicit form of the singlet projector in \eqref{eq:projector_P0} shows that $H_L$ annihilates them, so they are exact ground states at zero perturbation, and by Proposition~\ref{thm:spinwave_energies} their energies tend to $0$ as $K_{2s}, K_{2s-1} \to 0$. Below we quantify the closing of the gap, by perturbing the singlet model within the family \eqref{eq:general_interaction} and tracking the energy gap between the $D_L$-th and $(D_L+1)$-th eigenvalues, whose value at $0$ perturbation is the finite volume gap of the singlet model. 

We also note the instability we find implies the singlet model does not satisfy a local topological order condition of the sort required by the gap stability results of \cite{nachtergaele2022quasilocality}.


\subsubsection{Spin-$1$}
\label{sec:instability1}

For $s=1$ the family \eqref{eq:general_interaction} has a one dimensional parameter space after normalization,
\begin{equation}\label{eq:spin1_family}
    h_{x,x+1}(\theta) = \cos\theta\,\big(\vec{S}_x\cdot\vec{S}_{x+1}\big) + \sin\theta\,\big(\vec{S}_x\cdot\vec{S}_{x+1}\big)^2\,.
\end{equation}
The phase diagram in parameter space $\theta \in [0,2\pi)$ is shown below in Figure~\ref{fig:Spin1 Diagram}.

\begin{figure}[H]
  \centering
  \includegraphics[scale=0.4]{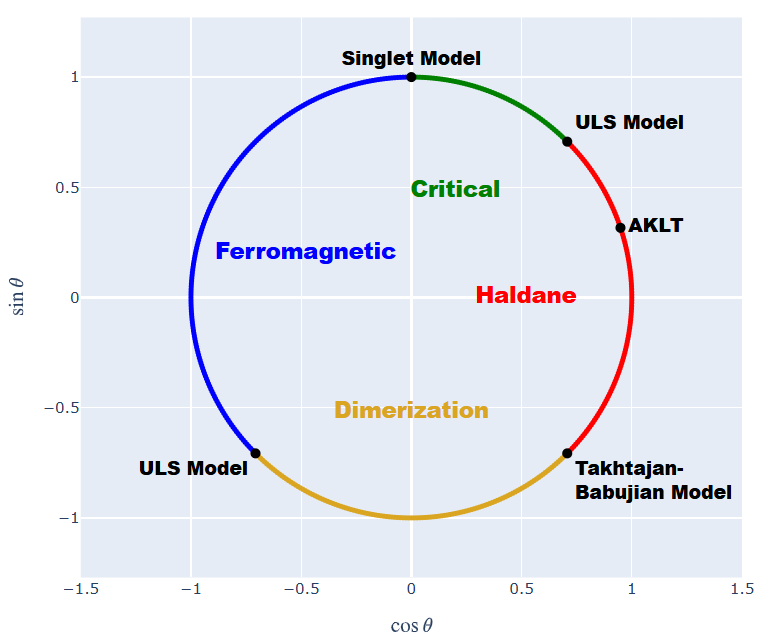}
  \caption{Conjectured phase diagram for spin-$1$ $SU(2)$ invariant chains. The singlet model is the purely biquadratic point up to a shift and normalization of the spectrum.}
  \label{fig:Spin1 Diagram}
\end{figure}

We perturb the singlet model into the two adjacent regions by setting
\begin{equation}\label{eq:H_1_eps}
    H_1(\epsilon, L) = \sum_x P^{(0)}_{x,x+1} - \epsilon\,P^{(1)}_{x,x+1}\,, \qquad \abs{\epsilon}<1\,.
\end{equation}
For $\epsilon<0$ this is a ferromagnetic chain as $P^{(2)}$ has the smallest coefficient in the interaction; for $\epsilon>0$ it enters the gapless phase described by $B_1$. Figure~\ref{fig:H_1_eps_spectrums} shows the spectrum of $H_1(\epsilon, L)$ for $L=4,5,6$ with open boundaries, and Figure~\ref{fig:perturbed_points_s1} labels these models in black within the spin-$1$ phase diagram. For $\epsilon$ bounded away from $0$ one sees a dense fan of states, spanned by the first $D_L$ eigenvectors, that collapse inward onto the $D_L$-fold degenerate ground state of the singlet model. We see a similar picture in the case of periodic boundary conditions.

\begin{figure}[H]
    \centering
    \begin{subfigure}[b]{\textwidth}
        \centering
        \includegraphics[width=\textwidth]{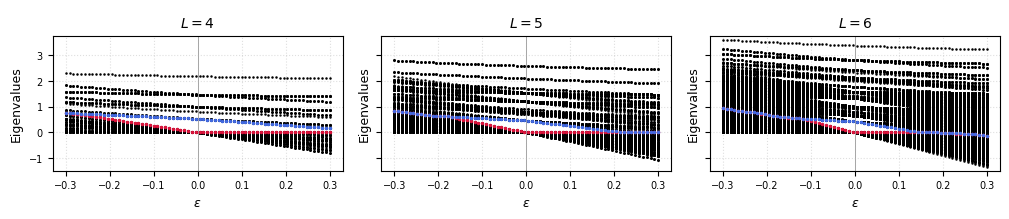}
    \end{subfigure}
    \caption{The spectrum of the Hamiltonian in \eqref{eq:H_1_eps} with $s=1$, OBC, $L=4,5,6$ plotted against $\epsilon$. The first $D_L$ eigenvalues are in red and the $(D_L+1)$-st eigenvalue is in blue.}
    \label{fig:H_1_eps_spectrums}
\end{figure}

\begin{figure}[H]
    \centering
    \begin{subfigure}[b]{0.5\textwidth}
        \centering
        \includegraphics[width=\textwidth]{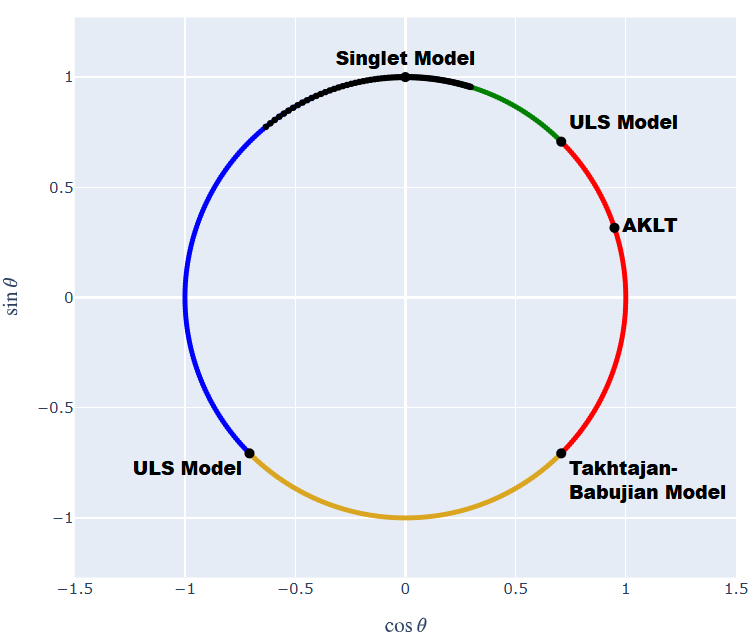}
    \end{subfigure}
    \caption{The perturbed models in the spin-$1$ phase diagram whose spectrum we plot in Figure~\ref{fig:H_1_eps_spectrums} are labeled in black.}
    \label{fig:perturbed_points_s1}
\end{figure}

We let $\zeta_1(\epsilon, L)$ denote the energy gap between the $D_L$-th and  $(D_L+1)$-st eigenvalue of the Hamiltonian in \eqref{eq:H_1_eps}. As stated earlier, this is expected to vanish as $L \to \infty$ for each fixed $\epsilon \neq 0$. Let 
\begin{equation}\label{eq:mu1}
\mu_1^\pm(L) = \lim_{\epsilon \to 0^\pm} \frac{\zeta_1(\epsilon, L) - \zeta_1(0, L)}{\epsilon}
\end{equation}
be the one-sided derivatives of the tracked level spacing for a particular $L$. The magnitude of $\mu_1^{\pm}(L)$ grows with $L$ in both directions of $\epsilon$, with the expectation that the rate at which the gap closes is extensive and hence that no $L$-independent lower bound on $\zeta_1(\epsilon, L)$ survives. Values of $\mu_1^\pm$ are computed for various $L$ in Table~\ref{tab:mu1_tab} and plotted in Figure~\ref{fig:mu1_fig} for open boundaries. We conjecture the magnitude of $\mu_1^\pm(L)$ grows linearly in $L$. We however see a slightly different picture in the case of periodic boundary conditions, with several points deviating too strongly for us to conclude confidently that the growth is linear. The values of $\mu_1^\pm$ in this case are computed for various $L$ and plotted below as well.

\begin{table}[h]
    \centering
    \small
    \begin{tabular}{ccccccccc}
        \hline
        $L$ & $D_L$ (OBC) & $\gamma_L$ (OBC) & $\mu_1^+$ (OBC) & $\mu_1^-$ (OBC) & $D_L$ (PBC) & $\gamma_L$ (PBC) & $\mu_1^+$ (PBC) & $\mu_1^-$ (PBC) \\
        \hline
        3 & 21 & 0.6667 & -0.500 & 1.500 & 18 & 0.6667 & -1.500 & 1.500 \\
        4 & 55 & 0.5286 & -1.158 & 1.868 & 47 & 0.3333 & -2.000 & 1.800 \\
        5 & 144 & 0.4607 & -2.000 & 2.587 & 123 & 0.3479 & -1.979 & 1.915 \\
        6 & 377 & 0.4226 & -2.713 & 3.455 & 322 & 0.3333 & -3.100 & 4.178 \\
        7 & 987 & 0.3994 & -3.568 & 4.146 & 843 & 0.3360 & -3.420 & 3.905 \\
        \hline
    \end{tabular}
    \caption{Degeneracies, gaps, and one-sided derivatives $\mu_1^\pm$ for $s=1$ and various system sizes $L$ comparing open (OBC) and periodic (PBC) boundary conditions.}
    \label{tab:mu1_tab}
\end{table}

\begin{figure}[H]
    \centering
    \begin{minipage}{0.48\textwidth}
        \centering
        \includegraphics[width=\textwidth]{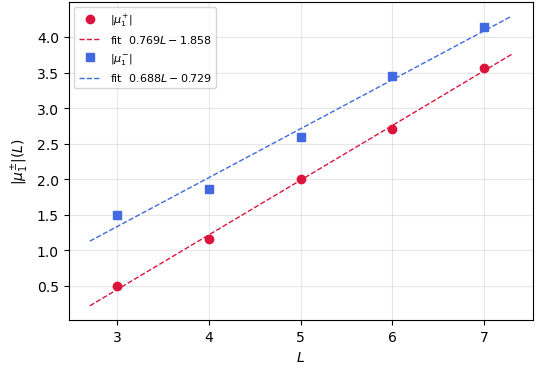}
    \end{minipage}\hfill
    \begin{minipage}{0.48\textwidth}
        \centering
        \includegraphics[width=\textwidth]{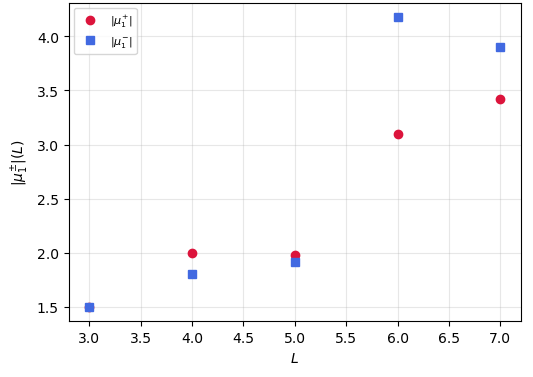}
    \end{minipage}
    \caption{Fitted plot of $\abs{\mu_1^\pm(L)}$ for $s=1$ under open boundary conditions (left) and periodic boundary conditions (right), for various system sizes $L$.}
    \label{fig:mu1_fig}
\end{figure}


\subsubsection{Spin-$\frac{3}{2}$}
\label{sec:instability32}

For $s=\tfrac32$ the family \eqref{eq:general_interaction} has a two dimensional parameter space after normalization,
\begin{equation}\label{eq:spin32_family}
    h_{x,x+1}(\theta, \phi) = \cos\theta\sin\phi\,\big(\vec{S}_x\cdot\vec{S}_{x+1}\big) + \sin\theta\sin\phi\,\big(\vec{S}_x\cdot\vec{S}_{x+1}\big)^2 + \cos\phi\,\big(\vec{S}_x\cdot\vec{S}_{x+1}\big)^3\,.
\end{equation}
The phase diagram in parameter space is now a sphere, and is partially colored as shown in Figure~\ref{fig:Spin32 Diagram}. The regions $A_{3/2} \setminus B_{3/2}$, $B_{3/2}$ and the ferromagnetic region are colored in analogy to the $s=1$ diagram. We color the rest of the sphere in grey to indicate those regions are quite unknown. We also label the singlet, ULS, $P^{(3)}$, $-P^{(1)}$ and $-P^{(2)}$ models. As mentioned previously, the ULS and $P^{(3)}$ models sit at the boundary of $A_{3/2}$. The $-P^{(1)}$ and $-P^{(2)}$ models sit at the boundary of $B_{3/2}$

\begin{figure}[H]
    \centering
    \includegraphics[width=0.85\textwidth]{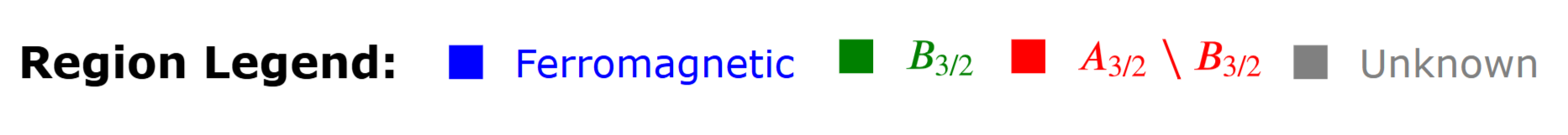}\\[1.5ex]

    \begin{minipage}[c]{0.48\textwidth}
        \centering
        \includegraphics[width=\textwidth]{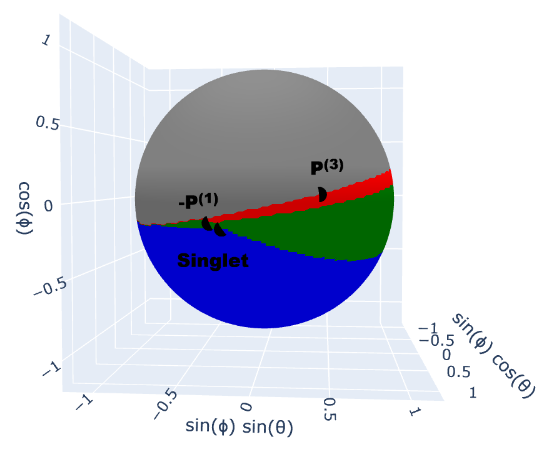}%
    \end{minipage}%
    \hfill
    \begin{minipage}[c]{0.48\textwidth}
        \centering
        \includegraphics[width=\textwidth]{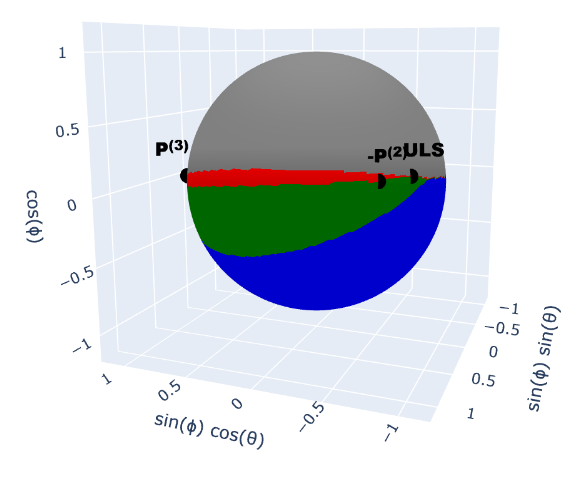}%
    \end{minipage}

    \caption{A conjectured partial phase diagram for spin-$\tfrac32$ $SU(2)$ invariant chains.}
    \label{fig:Spin32 Diagram}
\end{figure}

Fixing $K_0=1$ and $K_3=0$ after rewriting the spin-$\tfrac32$ interaction in terms of projectors, leaves a two parameter family of perturbations,
\begin{equation}\label{eq:H_32_eps,delt}
    H_{3/2}(\epsilon, \delta, L) = \sum_x P^{(0)}_{x,x+1}-\epsilon\,P^{(1)}_{x,x+1}-\delta\,P^{(2)}_{x,x+1}\,,
\end{equation}
and the same collapse of a dense fan of states into the singlet model's ground state is observed. We will focus on the case $\epsilon = \delta$ as this allows us to perturb both into the critical and ferromagnetic region according to our inequalities. Fixing either $\epsilon = 0$ or $\delta = 0$ and varying the other moves us along the boundary between these regions and not into the interior. The spectra of $H_{3/2}(\epsilon, \epsilon, L)$ are given in Figure~\ref{fig:H_eps,delt_specs} for various values of $\epsilon$ and $L=3,4,5$. The corresponding models are labeled by black dots within the spin-$3/2$ phase diagram in Figure~\ref{fig:perturbed_points_s32}. We again see a similar dense fan collapsing in the case of periodic boundary conditions.

\begin{figure}[H]
    \centering
    \begin{subfigure}[b]{\textwidth}
        \centering
        \includegraphics[width=\textwidth]{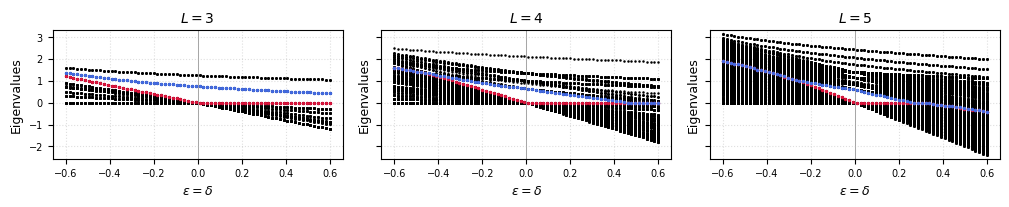}
    \end{subfigure}

    \caption{The spectrum of the Hamiltonian in \eqref{eq:H_32_eps,delt} with $s=\tfrac32$, open boundaries, $L=3,4,5$ plotted against $\epsilon = \delta$. The first $D_L$ eigenvalues are in red and the $(D_L+1)$-st eigenvalue is in blue.}
    \label{fig:H_eps,delt_specs}
\end{figure}

\begin{figure}[H]
    \centering
    \begin{minipage}[b]{0.48\textwidth}
        \centering
        \includegraphics[width=\textwidth]{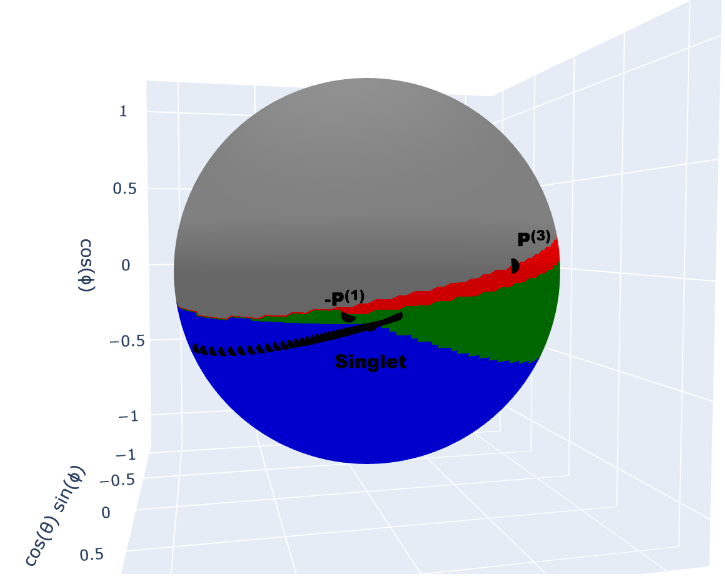}
    \end{minipage}
    \caption{The perturbed models in the spin-$\tfrac32$ phase diagram whose spectrum we plot in Figure~\ref{fig:H_eps,delt_specs} are labeled with black dots.}
    \label{fig:perturbed_points_s32}
\end{figure}

We let $\zeta_{3/2}(\epsilon, \delta, L)$ denote the energy gap between the $D_L$-th and $(D_L+1)$-st eigenvalue of the Hamiltonian in \eqref{eq:H_32_eps,delt}. We similarly let
\begin{equation}\label{eq:mu32}
\mu_{3/2}^\pm(L) = \lim_{\epsilon \to 0^\pm} \frac{\zeta_{3/2}(\epsilon, \epsilon, L) - \zeta_{3/2}(0, 0, L)}{\epsilon}
\end{equation}
be the one-sided derivatives of the tracked level spacing for a particular $L$. The magnitude of $\mu_{3/2}^{\pm}(L)$ again grows with $L$. Its values are computed for various $L$ in Table~\ref{tab:mu32_tab} and plotted in Figure~\ref{fig:mu32_fig} for the case of open and periodic boundary conditions. Similarly to the $s=1$ case, the fit in $L$ is linear for open boundaries and not exactly linear for periodic boundary.

\begin{table}[h]
    \centering
    \small
    \begin{tabular}{ccccccccc}
        \hline
        $L$ & $D_L$ (OBC) & $\gamma_L$ (OBC) & $\mu_{3/2}^+$ (OBC) & $\mu_{3/2}^-$ (OBC) & $D_L$ (PBC) & $\gamma_L$ (PBC) & $\mu_{3/2}^+$ (PBC) & $\mu_{3/2}^-$ (PBC) \\
        \hline
        3 & 56 & 0.7500 & -0.667 & 1.333 & 52 & 0.5000 & -2.500 & 0.500 \\
        4 & 209 & 0.6464 & -1.541 & 1.878 & 194 & 0.5000 & -2.800 & 2.425 \\
        5 & 780 & 0.5955 & -2.401 & 2.616 & 724 & 0.5000 & -3.200 & 2.944 \\
        6 & 2911 & 0.5670 & -3.299 & 3.471 & 2702 & 0.5000 & -3.956 & 4.176 \\
        \hline
    \end{tabular}
    \caption{Degeneracies, gaps, and one-sided derivatives $\mu_{3/2}^\pm$ for $s=3/2$ and various system sizes $L$ comparing open (OBC) and periodic (PBC) boundary conditions.}
    \label{tab:mu32_tab}
\end{table}

\begin{figure}[H]
    \centering
    \begin{minipage}[b]{0.48\textwidth}
        \centering
        \includegraphics[width=\textwidth]{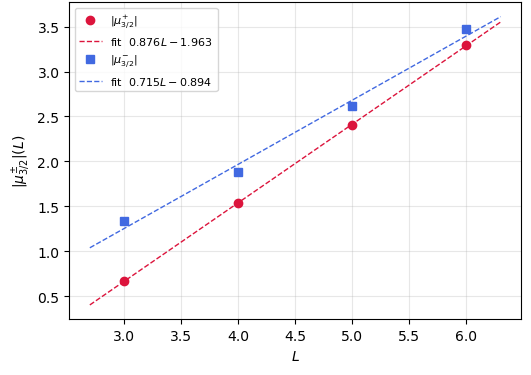}
    \end{minipage}
    \hfill
    \begin{minipage}[b]{0.48\textwidth}
        \centering
        \includegraphics[width=\textwidth]{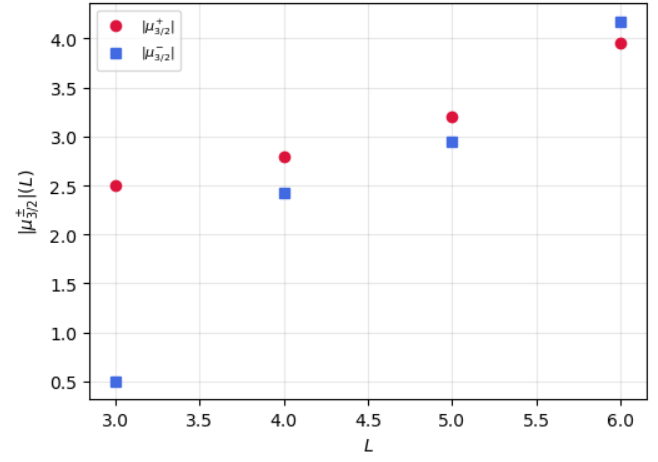}
    \end{minipage}
    \caption{Fitted plot of $\abs{\mu_{3/2}^\pm(L)}$ for $s=\tfrac32$ under open boundary conditions (left) and periodic boundary conditions (right), for various system sizes $L$.}
    \label{fig:mu32_fig}
\end{figure}

\begin{appendix}

\section{Spin Waves for $SU(2)$ Invariant Interactions}
\begin{prop}\label{thm:spinwave_energies}
Consider the general $SU(2)$ invariant two-site interaction $h_{x,x+1}$ as in Equation~\eqref{eq:general_interaction}. The state
\begin{equation}\label{eq:psi_k}
    \psi_k = \sum_{y=1}^L e^{iky} \ket{s-1}_y
\end{equation}
is an exact eigenstate of the Hamiltonian $H_{SU(2), L}^{\text{PBC}} = \sum_{x=1}^L h_{x,x+1}$ under periodic boundary conditions, with corresponding energy
\begin{equation}\label{eq:E_k}
    E_k = K_{2s}(L-2) + K_{2s-1}(1-\cos k) + K_{2s}(1+\cos k)\,,
\end{equation}
where $k=\frac{2\pi}{L}j$ for $j = 0,1,\dots,L-1$.
\end{prop}
\begin{proof}
By the structure of the Clebsch-Gordan coefficients, the only projectors $P^{(i)}$ that do not annihilate $\ket{s, s-1}$ and $\ket{s-1, s}$ are those with $i=2s$ and $i=2s-1$. In such cases, reading off the Clebsch-Gordan coefficients gives
\begin{align}\label{eq:clebsch_gordan}
    P^{(2s)} \ket{s, s-1} &= \frac{1}{2} \left[\ket{s, s-1} + \ket{s-1, s}\right]\,, \nonumber \\
    P^{(2s-1)} \ket{s, s-1} &= \frac{1}{2} \left[\ket{s, s-1} - \ket{s-1, s}\right]\,, \nonumber \\
    P^{(2s)} \ket{s-1, s} &= \frac{1}{2}\left[ \ket{s-1, s} + \ket{s, s-1}\right]\,, \nonumber \\
    P^{(2s-1)} \ket{s-1, s} &= \frac{1}{2}\left[\ket{s-1, s} - \ket{s, s-1}\right]\,.
\end{align}
All other projectors act uniformly as $0$. Thus, for each $y \in [1,L]$ we have:
\begin{align}\label{eq:H_action_state_y2}
    H_{SU(2), L}^{\text{PBC}} \ket{s-1}_y &= \sum_{x=1}^L \left(K_{2s} P^{(2s)}_{x,x+1} + K_{2s-1} P^{(2s-1)}_{x,x+1}\right) \ket{s-1}_y \nonumber \\
    &= K_{2s}(L-2) \ket{s-1}_y + \frac{K_{2s}}{2}\left[\ket{s-1}_{y-1} + \ket{s-1}_y\right] + \frac{K_{2s-1}}{2}\left[\ket{s-1}_y - \ket{s-1}_{y-1}\right] \nonumber \\
    &\quad + \frac{K_{2s}}{2}\left[\ket{s-1}_y + \ket{s-1}_{y+1}\right] + \frac{K_{2s-1}}{2}\left[\ket{s-1}_y - \ket{s-1}_{y+1}\right]\,.
\end{align}
Applying $H_{SU(2), L}^{\text{PBC}}$ directly to $\psi_k$ and substituting the derived expression enables us to factor out $e^{\pm iky}$ by shifting indices within the sums over $y-1$ and $y+1$. Recombining gives $E_k\psi_k$. As usual, translation invariance of $\psi_k$ discretizes the momenta as $k = \frac{2\pi}{L}j$ for $j = 0,1,\dots,L-1$.
\end{proof}

For the spin-$1$ case, the interaction is of the form $\cos \theta \vec{S}\cdot\vec{S} + \sin \theta(\vec{S}\cdot \vec{S})^2$ as discussed in the introduction and Section~\ref{sec:instability1}. We plot in Figure~\ref{fig:Spin Wave Spec} the spectrum of the Hamiltonian on a periodic chain of length $L= 5$ as a function of $\theta$ and label in red those eigenvalues corresponding to the spin wave energies we just found. We see these energies minimize the spectrum in the ferromagnetic region, maximize the spectrum in the antiferromagnetic region, and play intermediary roles in other regions. 

\begin{figure}[H]
    \centering
    \includegraphics[width=\textwidth]{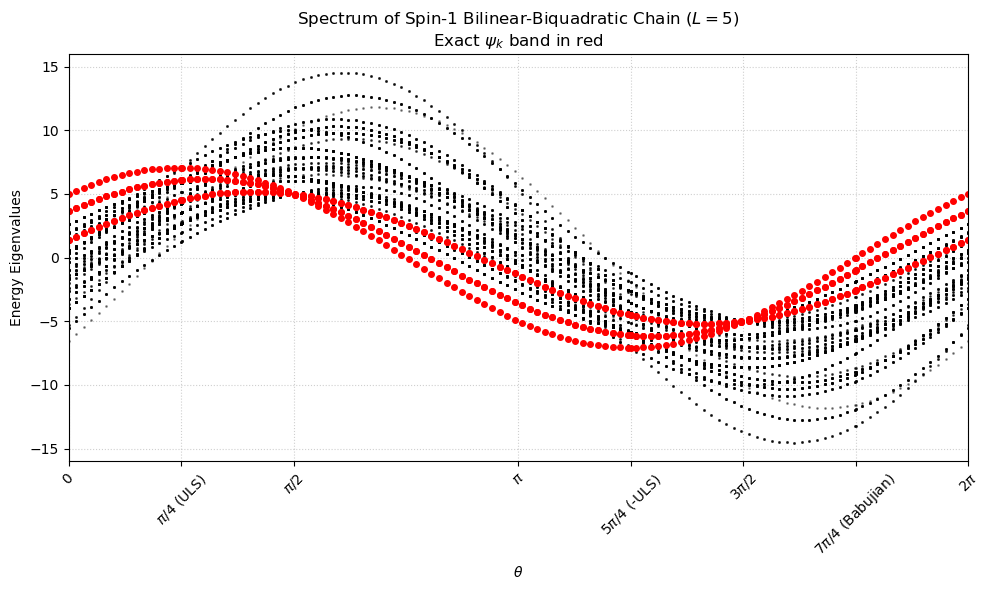} 
    \caption{Spectrum of the Hamiltonian for the bilinear-biquadratic interaction on a periodic chain of length $L=5$ with spin wave energies highlighted in red.}
    \label{fig:Spin Wave Spec}
\end{figure}


\section{Convergence of Dynamics and the Spectrum}
\label{sec:spin_system_facts}

Recall $\tau_t^{\Phi_n}$ and $\tau_t$ denote the infinite volume dynamics, and $\tau_t^{\Phi_n, \Lambda}$ and $\tau_t^\Lambda$ denote the finite volume dynamics, associated to the interactions $\Phi_n$ and $\Phi$ as before. Let $f \in \mathcal{S}(\rr)$ be a Schwartz function and define $\tau_f(A) = \int f(t)\tau_t(A)\,\dd{t}$ for each $A \in \cA_{\text{loc}}$. Similarly define $\tau_f^{\Phi_n}(A)$. 

\begin{lem}\label{lem:dynamics_convg}
For all $A \in \cA_{\text{loc}}$ and $T \geq 0$,
\begin{equation}\label{eq:dynamics_convg}
    \lim_{n \to \infty} \,\sup_{\abs{t}\leq T} \big\|\tau_t^{\Phi_n}(A) - \tau_t(A)\big\| = 0
\end{equation}
and
\begin{equation}\label{eq:smeared_dynamics_convg}
    \lim_{n \to \infty} \big\|\tau_f(A) - \tau_f^{\Phi_n}(A)\big\| = 0\,.
\end{equation}
\end{lem}
\begin{proof}
Let $X=\mathrm{supp}(A)$ and consider $n$ large enough so that $X \subset \Lambda_n=[-L_n,L_n]$. Define $g_n(r) = \tau_r^{\Phi_n,\Lambda_n} \bigl(\tau_{t-r}^{\Lambda_n}(A)\bigr)$. Differentiating gives
\begin{align*}
g_n'(r) & = -i\,\tau_r^{\Phi_n,\Lambda_n}\left([H_{\Lambda_n}^\Phi,\tau_{t-r}^{\Lambda_n}(A)]\right) + i\,\tau_r^{\Phi_n,\Lambda_n} \left( [H_{\Lambda_n}^{\Phi_n},\tau_{t-r}^{\Lambda_n}(A)] \right)\\
&= i\,\tau_r^{\Phi_n,\Lambda_n}\left([P^{(0)}_{L_n,-L_n},\tau_{t-r}^{\Lambda_n}(A)]\right)\,.
\end{align*}
Integration yields
\[
\tau_t^{\Phi_n,\Lambda_n}(A)-\tau_t^{\Lambda_n}(A)
= \int_0^t i\,\tau_r^{\Phi_n,\Lambda_n}\left([P^{(0)}_{L_n,-L_n},\tau_{t-r}^{\Lambda_n}(A)]\right)\dd{r}\,.
\]
Making substitutions based on $t>0$ or $t<0$ gives, for all $t\in\rr$,
\[
\bigl\|\tau_t^{\Phi_n,\Lambda_n}(A)-\tau_t^{\Lambda_n}(A)\bigr\| \leq \int_0^{|t|} \bigl\|[\tau_{\mathrm{sgn}(t)s}^{\Lambda_n}(A),P^{(0)}_{L_n,-L_n}]\bigr\| \dd{s}\,.
\]
Let $Y_n = \{-L_n, L_n\}$ so for all $n$ sufficiently large, $X\cap Y_n=\varnothing$. The Lieb--Robinson bound of Theorem~3.1 in \cite{nachtergaele2019quasilocality} gives
\[
\bigl\|[\tau_{\mathrm{sgn}(t)s}^{\Lambda_n}(A),P^{(0)}_{L_n,-L_n}]\bigr\| \leq \frac{2\|A\|}{C_F}\left(e^{2C_F|s|\,\|\Phi\|_F}-1\right) \sum_{x\in X}\sum_{y\in Y_n}F(|x-y|)\,.
\]
Hence, for $|t|\leq T$,
\[
\bigl\|\tau_t^{\Phi_n,\Lambda_n}(A)-\tau_t^{\Lambda_n}(A) \bigr\| \leq \frac{2\|A\|}{C_F}\,T\left(e^{2C_FT\|\Phi\|_F}-1\right) \sum_{x\in X}\sum_{y\in Y_n}F(|x-y|).
\]
The double sum has $2\abs{X}$ terms, each of the form $F(\abs{x-y})$ for $x \in X$ and $y \in \{-L_n, L_n\}$. Since $F(z)\to0$ as $z\to\infty$, the sum tends to zero. Therefore
\[
\lim_{n\to\infty}\sup_{|t|\leq T}\bigl\|\tau_t^{\Phi_n,\Lambda_n}(A)-\tau_t^{\Lambda_n}(A)\bigr\|=0\,.
\]
On the other hand, the usual thermodynamic-limit convergence for the interaction $\Phi$, Theorem~3.5 in \cite{nachtergaele2019quasilocality}, gives
\[
\lim_{n\to\infty}\sup_{|t|\leq T}\bigl\|\tau_t^{\Lambda_n}(A)-\tau_t(A)\bigr\|=0\,.
\]
Since $\Phi_n$ has no interaction terms outside $\Lambda_n$,  $\tau_t^{\Phi_n}(A)=\tau_t^{\Phi_n,\Lambda_n}(A)$ for all sufficiently large $n$. The triangle inequality now gives the first part of the statement in Equation~\eqref{eq:dynamics_convg}.

It remains to pass from $\tau_t$ to $\tau_f$. Let $f \in \mathcal{S}(\rr)$ and $\epsilon > 0$. Since $\|\tau_t(A)-\tau^{\Phi_n}_t(A)\| \leq 2\|A\|$ for every $t$ and $n$, splitting the integral at $|t| = T$ gives
\[
    \big\|\tau_f(A)-\tau^{\Phi_n}_f(A)\big\|
    \;\leq\; \int_{|t| \leq T}|f(t)|\,\big\|\tau_t(A)-\tau^{\Phi_n}_t(A)\big\|\,\dd t
    \;+\; 2\|A\|\int_{|t| > T}|f(t)|\,\dd t\,.
\]
By the rapid decay of $f$ choose $T$ so that the second term is less than $\epsilon/2$; the first term is at most $\|f\|_{L^1}\sup_{|t|\leq T}\|\tau_t(A)-\tau^{\Phi_n}_t(A)\|$, which is less than $\epsilon/2$ for all large $n$ by Equation~\eqref{eq:dynamics_convg}. As $\epsilon$ was arbitrary, the second claim in Equation~\eqref{eq:smeared_dynamics_convg} follows.
\end{proof}

In our argument for Theorem~\ref{thm:gns_gap}, we use a lower bound on the finite volume gaps when periodic boundary conditions are considered. We need this to transfer to a lower bound for the gap of $H$, where $H$ is constructed using the interaction implementing open boundaries. To this end we use a statement from Proposition 5.4 of \cite{bachmann2016lieb}, adapted to the present setting, that says the spectral gaps of the finite volume Hamiltonians persist in the GNS representation.

\begin{prop}\label{prop:spectrum_transfer}
Let $H$ be the GNS Hamiltonian of $\omega$ with respect to $\tau_t$ as before. Suppose $E \in \rr$ is such that there exist $\epsilon > 0$ and $N \in \nn$ with
\begin{equation}\label{eq:gap_hypothesis}
    (E-\epsilon,\,E+\epsilon) \cap \sigma\big(H^{\Phi_n}_{\Lambda_n}\big) = \emptyset
    \qquad \text{for all } n \geq N\,,
\end{equation}
where $\Lambda_n = [-L_n, L_n]$ as before. Then $E \notin \sigma(H)$.
\end{prop}
\begin{proof}
Let $f \in \mathcal{S}(\rr)$ be a Schwartz function whose Fourier transform $\hat f \in C_c^\infty((E-\epsilon, E+\epsilon))$ with $\hat f(E)=1$. It suffices to show that $\langle \pi(A)\Omega, \hat f(H)\pi(B)\Omega\rangle = 0$ for all $A,B \in \cA_{\text{loc}}$. Setting $\tau_f(B) = \int f(t)\tau_t(B)\dd{t}$ as before and using invariance of $\Omega$,
\[
    \hat f(H)\pi(B)\Omega = \int f(t)\,e^{itH}\pi(B)e^{-itH}\Omega\,\dd t = \pi(\tau_f(B))\Omega\,,
\]
so that $\langle \pi(A)\Omega, \hat f(H)\pi(B)\Omega\rangle = \omega(A^*\tau_f(B))$.

Fix $n$ with $n \geq N$ and $\Lambda_n \supseteq \mathrm{supp}\,A \cup \mathrm{supp}\,B$. The interaction $\Phi_n$ is nonzero only on the $N_n$ bonds of the ring on $\Lambda_n$, so $H^{\Phi_n}_{\Lambda_n}$ is a bounded element of $\cA_{\Lambda_n}$ and, for every $\Lambda \supseteq \Lambda_n$, $H^{\Phi_n}_\Lambda = H^{\Phi_n}_{\Lambda_n}$. Hence $\tau_t^{\Phi_n}$ is implemented by $H^{\Phi_n}_{\Lambda_n}$, preserves $\cA_{\Lambda_n}$, and
$\tau_f^{\Phi_n}(B) \in \cA_{\Lambda_n}$. As the family $\{\omega_\Lambda\}$ is consistent, $\omega$ restricts on $\cA_{\Lambda_n}$ to $\omega_{\Lambda_n}$, whence
\[
    \omega\big(A^*\tau_f^{\Phi_n}(B)\big)
    = \omega_{\Lambda_n}\big(A^*\tau_f^{\Phi_n}(B)\big)
    = \omega_{\Lambda_n}\Big(A^*\,\hat f\big(H^{\Phi_n}_{\Lambda_n}\big)\,B\Big)\,,
\]
the second equality because $H^{\Phi_n}_{\Lambda_n}\ket{s\cdots s} = 0$ by frustration freeness, so that $e^{-itH^{\Phi_n}_{\Lambda_n}}\ket{s\cdots s} = \ket{s\cdots s}$ and the $t$-integral collects into $\hat f\big(H^{\Phi_n}_{\Lambda_n}\big)$. By \eqref{eq:gap_hypothesis} the function $\hat f$ is supported outside $\sigma\big(H^{\Phi_n}_{\Lambda_n}\big)$, so $\hat f\big(H^{\Phi_n}_{\Lambda_n}\big) = 0$ and the expression vanishes. Therefore $\omega\big(A^*\tau_f^{\Phi_n}(B)\big) = 0$ for every such $n$, and since
\[
    \big|\omega\big(A^*\tau_f(B)\big)\big|
    = \big|\omega\big(A^*\tau_f(B)\big) - \omega\big(A^*\tau_f^{\Phi_n}(B)\big)\big|
    \leq \|A\|\,\big\|\tau_f(B)-\tau_f^{\Phi_n}(B)\big\|\,,
\]
Lemma~\ref{lem:dynamics_convg} gives $\omega\big(A^*\tau_f(B)\big) = 0$ on letting $n \to \infty$.
\end{proof}


\section{Module Theory and the Double Centralizer Theorem}

Here we list the general algebra facts we use throughout the paper. All algebras of interest are semisimple and over $\cc$. We begin with a fact about faithful modules. The faithfulness assumption in Theorem~\ref{thm:double_centralizer} is to have all distinct simple $A$-modules occur in the decomposition. The faithfulness of $\rho_d$ in Proposition~\ref{prop:decomposition} served this case as well: each $W_{L,k}$ occurs in the decomposition. This is in general given by the following simple consequence of the Artin-Wedderburn theorem, which can be found in standard texts \cite{cohn2003basic}, \cite{lang2002algebra}.

\begin{prop}\label{prop:faithful_iff}
Let $A$ be a finite dimensional semisimple algebra with $\cH_1,\dots,\cH_r$ a complete list of its simple modules up to isomorphism, and
\[
    \cH \;\cong\; \bigoplus_{k=1}^{r} \cH_k^{\oplus m_k}
\]
a finite dimensional $A$-module. Then $\cH$ is faithful if and only if $m_k > 0$ for every $k$.
\end{prop}
\begin{proof}
By the Artin-Wedderburn theorem, $A = \bigoplus_{k=1}^{r}A_k$ with each $A_k \cong \rmEnd(\cH_k)$ the simple two-sided ideals, indexed so that $A_k$ acts faithfully on $\cH_k$ and annihilates $\cH_j$ for $j \neq k$; moreover every two-sided ideal of $A$ is a sum of the $A_k$.

Since $\mathrm{Ann}(\cH)$ is a two-sided ideal, it is therefore $\bigoplus_{j \in S}A_j$ for some $S \subseteq \{1,\dots,r\}$, and $A_j \subseteq \mathrm{Ann}(\cH)$ precisely when $A_j$ annihilates every simple constituent of $\cH$. As $A_j$ annihilates $\cH_k$ for $k \neq j$ and acts faithfully on $\cH_j$, this happens precisely when $\cH_j$ does not occur in $\cH$, that is when $m_j = 0$. Hence
\[
    \mathrm{Ann}(\cH) \;=\; \bigoplus_{j\,:\,m_j = 0} A_j\,,
\]
and since each $A_j \neq 0$, the annihilator vanishes if and only if every $m_k$ is strictly positive.
\end{proof}

The next important theorem is the double centralizer theorem, which was the basis of our decomposition Proposition~\ref{prop:decomposition}. The proofs can be found in \cite{procesi2007lie} and \cite{goodman2009symmetry}. As stated earlier, the assumption of faithfulness guarantees each standard module $W_k$ occurs in the decomposition. Without this assumption, the decomposition still holds, with some of the multiplicity spaces vanishing. 

\begin{thm}[Double Centralizer Theorem]\label{thm:double_centralizer}
Let $A$ be a semisimple algebra and $\cH$ a finite-dimensional faithful $A$-module. Let $B = \rmEnd_A(\cH)$ be the centralizer of $A$ in $\rmEnd(\cH)$. Then:
\begin{enumerate}
    \item $B$ is a semisimple algebra.
    \item The centralizer of $B$ in $\rmEnd(\cH)$ is exactly $A$, i.e., $\rmEnd_B(\cH) = A$.
    \item As an $A \otimes B$ module, $\cH$ decomposes as 
    \[ \cH \cong \bigoplus_k W_k \otimes U_k \]
    where the $W_k$ are pairwise non-isomorphic simple $A$-modules and the $U_k$ are pairwise non-isomorphic simple $B$-modules. The sum runs over all distinct simple modules of $A$, which are in bijection with those of $B$.
\end{enumerate}
\end{thm}

The following is elementary and requires neither semisimplicity nor faithfulness, but it is convenient to have it stated alongside the results that do.

\begin{lem}\label{lem:iso_spectrum}
Let $A$ be an algebra and let $V,V'$ be isomorphic $A$-modules, with $\pi$ and $\pi'$ the corresponding actions. Then $\sigma\big(\pi(a)\big) = \sigma\big(\pi'(a)\big)$ for every $a \in A$.
\end{lem}
\begin{proof}
Let $\Theta : V \to V'$ be an isomorphism of $A$-modules, so that $\Theta(a\cdot v) = a\cdot\Theta(v)$ for all $a \in A$ and $v \in V$. This reads $\Theta\pi(a) = \pi'(a)\Theta$, whence $\pi'(a) = \Theta\pi(a)\Theta^{-1}$, and conjugate operators have the same eigenvalues.
\end{proof}

Lemma~\ref{lem:iso_spectrum} and Corollary~\ref{cor:isospectral_general} below are complementary, and neither implies the other. The first compares two copies of the same module and applies in particular to simple modules, which are never faithful once $A$ has more than one block. The second compares two faithful modules, which are typically not isomorphic, and needs semisimplicity to reduce to a common list of simple constituents.

\begin{prop}\label{prop:isospectral_general}
Let $A$ be a finite dimensional semisimple algebra and let $
V,V'$ be $A$-modules with $V'$ faithful. If $a \in A$ and $0 \neq v \in V$ satisfy $a\cdot v = \lambda v$ for some $\lambda \in \cc$, then there exists $0 \neq v' \in V'$ with $a\cdot v' = \lambda v'$.
\end{prop}
\begin{proof}
Let $\cH_1,..., \cH_r$ be a complete list of simple $A$ modules up to isomorphism. Since $A$ is semisimple, so is every $A$-module, and we may fix isomorphisms of $A$-modules
\[
    \theta : V \longrightarrow \bigoplus_{i=1}^{r}\cH_i^{\oplus m_i}\,,
    \qquad
    \theta' : \bigoplus_{i=1}^{r}\cH_i^{\oplus m_i'} \longrightarrow V'\,,
\]
where $m_i$ and $m_i'$ are the multiplicities of $\cH_i$ in $V$ and in $V'$ respectively. Because $V'$ is faithful, Proposition~\ref{prop:faithful_iff} gives $m_i' \geq 1$ for every $i$. For $1 \leq j \leq m_i$ we write $\cH_i^{(j)}$ for the $j$-th copy of $\cH_i$ in the decomposition of $V$, and similarly for the copy in the domain of $V'$.

Since $\theta$ intertwines the actions, $\theta(v) \neq 0$ and
$a\cdot\theta(v) = \theta(a\cdot v) = \lambda\,\theta(v)$. Decompose
\[
    \theta(v) = \sum_{i=1}^{r}\sum_{j=1}^{m_i}w_{i,j}\,, \qquad w_{i,j} \in \cH_i^{(j)}\,.
\]
Each $\cH_i^{(j)}$ is a submodule and hence invariant under the action of $a$, so
$a\cdot w_{i,j}-\lambda w_{i,j} \in \cH_i^{(j)}$ and
\[
    0 \;=\; a\cdot\theta(v)-\lambda\,\theta(v) \;=\; \sum_{i,j}\big(a\cdot w_{i,j}-\lambda w_{i,j}\big)\,.
\]
Since we have a direct sum, every term therefore vanishes, so $a\cdot w_{i,j} = \lambda w_{i,j}$ for all $i,j$. As $\theta(v) \neq 0$, some $w_{i,j} \neq 0$; fix such a pair and regard $w \coloneq w_{i,j}$ as a nonzero element of $\cH_i$ satisfying $a\cdot w = \lambda w$.

Since $m_i' \geq 1$, $V'$ has a copy of $\cH_i$. Let $\iota : \cH_i \to \bigoplus_{l}\cH_l^{\oplus m_l'} \cong V'$ be the inclusion of that copy and set $v' \coloneq \theta'\big(\iota(w)\big) \in V'$. Then $v' \neq 0$ because $\theta'$ and $\iota$ are injective and $w \neq 0$, and since both intertwine the actions,
\[
    a\cdot v' = a\cdot\theta'\big(\iota(w)\big) = \theta'\big(\iota(a\cdot w)\big)
    = \theta'\big(\iota(\lambda w)\big) = \lambda\,v'\,. \qedhere
\]
\end{proof}

\begin{cor}\label{cor:isospectral_general}
Let $A$ be a finite dimensional semisimple $\cc$-algebra and let $V,V'$ be faithful $A$-modules. Then for every $a \in A$ the operators given by the action of $a$ on $V$ and on $V'$ have the same set of eigenvalues.
\end{cor}

\end{appendix}

\paragraph{AI Statement}
The authors acknowledge the use of Claude (Anthropic) for editorial assistance, for help working out the details of several arguments, and for independent numerical checks. Google Gemini was used for assistance with the code generating the tables and figures. Full responsibility for the entire content rests with the authors.

\paragraph{Funding information}
This work was supported in part by the National Science Foundation under grant DMS-2510824 (BN \& RF). BN would like to thank the Isaac Newton Institute for Mathematical Sciences, Cambridge, for support and hospitality during the programme Mathematics of Many-Body Entanglement, where work on this paper was undertaken. This work was partially supported by EPSRC grant EP/Z000580/1 and by a grant from the Simons Foundation (BN).

\bibliographystyle{amsalpha}
\bibliography{references}

\end{document}